\documentclass[acmsmall]{acmart}

\usepackage{array}
\usepackage{mathtools, amsmath, amsthm, amsfonts}
\usepackage{xspace}
\usepackage{graphicx}
\usepackage{algorithm2e}
\usepackage{hyperref}
\usepackage[capitalise]{cleveref}
\usepackage{xcolor}
\usepackage{tikz}
\usetikzlibrary{arrows,automata}
\usepackage{todonotes}

\usepackage[appendix=append]{apxproof} 

\newcommand{\mathcommand}[1]{\textnormal{\ensuremath{{#1}}}\xspace}

\newcommand*{\sfname}[1]{\mathcommand{\textsf{#1}}}
\newcommand*{\scname}[1]{\mathcommand{\textsc{#1}}}
\newcommand*{\calname}[1]{\mathcommand{\mathcal{#1}}}
\newcommand*{\bbname}[1]{\mathcommand{\mathbb{#1}}}
\newcommand*{\ttname}[1]{\mathcommand{\mathtt{#1}}}

\makeatletter
\newcommand*{\set}[1]{%
  \relax\if@display
    \mathcommand{\left\{#1\right\}}
  \else
    \mathcommand{\{ #1 \}}
  \fi}
\newcommand*{\setst}[2]{%
  \relax\if@display
    \mathcommand{\left\{ #1 \;\middle|\; #2 \right\}}
  \else
    \mathcommand{\{ #1 \;|\; #2 \}}
  \fi}
\makeatother
\newcommand*{\size}[1]{\mathcommand{\left|#1\right|}}

\newcommand*{\nat}{\bbname{N}}

\theoremstyle{plain}
\newtheoremrep{theorem}{Theorem}
\newtheoremrep{definition}[theorem]{Definition}
\newtheoremrep{proposition}[theorem]{Proposition}
\newtheoremrep{lemma}[theorem]{Lemma}
\newtheoremrep{conjecture}[theorem]{Conjecture}
\newtheoremrep{corollary}[theorem]{Corollary}
\newtheoremrep{example}[theorem]{Example}

\newcommand*{\re}{\scname{RE}}

\newcommand*{\qinit}{\mathcommand{q_0}}

\newcommand*{\qloop}{\mathcommand{q_r}} 

\newcommand*{\blank}{\ttname{B}}
\newcommand*{\goleft}{\mathcommand{\mathop{\leftarrow}}}
\newcommand*{\goright}{\mathcommand{\mathop{\rightarrow}}}

\newcommand*{\setR}{\calname{R}}

\newcommand*{\kb}{\calname{K}}
\newcommand*{\Preds}{\sfname{Preds}}
\newcommand*{\Vars}{\sfname{Vars}}
\newcommand*{\Consts}{\sfname{Cons}}
\newcommand*{\Nulls}{\sfname{Nulls}}
\renewcommand*{\terms}{\sfname{terms}}
\newcommand*{\ar}{\sfname{ar}}

\newcommand*{\body}{\sfname{body}}
\newcommand*{\head}{\sfname{head}}

\newcommand*{\trigoutput}{\sfname{output}}
\newcommand*{\trigsupport}{\sfname{support}}
\newcommand*{\der}{\calname{D}}

\newcommand*{\angles}[1]{\mathcommand{\langle #1 \rangle}}

\newcommand*{\arule}{\ensuremath{R}\xspace}
\newcommand*{\aquery}{\ensuremath{Q}\xspace}
\newcommand*{\subs}{\ensuremath{\sigma}\xspace}
\newcommand*{\aruleset}{\ensuremath{\Sigma}\xspace}
\newcommand*{\db}{\ensuremath{D}\xspace}
\newcommand*{\KB}{\calname{K}}
\newcommand*{\chaseseq}{\calname{F}}

\newcommand*{\conf}{\sfname{conf}}

\renewcommand*{\terms}{\sfname{terms}}

\newcommand*{\vx}{\mathcommand{\vec{x}}}
\newcommand*{\vy}{\mathcommand{\vec{y}}}
\newcommand*{\vz}{\mathcommand{\vec{z}}}
\newcommand*{\vt}{\mathcommand{\vec{t}}}
\newcommand*{\brake}{\predname{brSet}}

\newcommand*{\CT}{\sfname{CT}}

\newcommand*{\CTR}[1]{\ensuremath{\CT^\textit{rest}_{#1}}\xspace}

\newcommand*{\CTda}{\ensuremath{\CTR{\db\forall}}\xspace}

\newcommand*{\CTaa}{\ensuremath{\CTR{\forall\forall}}\xspace}

\newcommand*{\TerminatingRuleSets}[2]{\ensuremath{\sfname{CTR}_{#1}^\textit{#2}}\xspace}
\newcommand*{\TerminatingRuleSetsOblivious}[1]{\TerminatingRuleSets{#1}{obl}}
\newcommand*{\TRO}[1]{\TerminatingRuleSetsOblivious{#1}}
\newcommand*{\TerminatingRuleSetsRestricted}[1]{\TerminatingRuleSets{#1}{rest}}
\newcommand*{\TRR}[1]{\TerminatingRuleSetsRestricted{#1}}
\newcommand*{\TerminatingRuleSetsRestrictedBreadthFirst}[1]{\TerminatingRuleSets{#1}{b1r}}
\newcommand*{\TRRBF}[1]{\TerminatingRuleSetsRestrictedBreadthFirst{#1}}
\newcommand*{\TerminatingRuleSetsCore}[1]{\TerminatingRuleSets{#1}{core}}
\newcommand*{\TRC}[1]{\TerminatingRuleSetsCore{#1}}

\newcommand*{\TerminatingKBs}[2]{\ensuremath{\sfname{CTK}_{#1}^\textit{#2}}\xspace}
\newcommand*{\TerminatingKBsOblivious}[1]{\TerminatingKBs{#1}{obl}}
\newcommand*{\TKO}[1]{\TerminatingKBsOblivious{#1}}
\newcommand*{\TerminatingKBsRestricted}[1]{\TerminatingKBs{#1}{rest}}
\newcommand*{\TKR}[1]{\TerminatingKBsRestricted{#1}}
\newcommand*{\TerminatingKBsRestrictedBreadthFirst}[1]{\TerminatingKBs{#1}{b1r}}
\newcommand*{\TKRBF}[1]{\TerminatingKBsRestrictedBreadthFirst{#1}}
\newcommand*{\TerminatingKBsCore}[1]{\TerminatingKBs{#1}{core}}
\newcommand*{\TKC}[1]{\TerminatingKBsCore{#1}}

\newcommand*{\predname}[1]{\ensuremath{\mathtt{#1}}\xspace}
\newcommand*{\pn}[1]{\predname{#1}}
\newcommand*{\preda}{\predname{a}}
\newcommand*{\predb}{\predname{b}}
\newcommand*{\predc}{\predname{c}}
\newcommand*{\predq}{\predname{q}}
\newcommand*{\predqp}{\predname{q'}}
\newcommand*{\predbrake}{\predname{Brake}}
\newcommand*{\predqstart}{\ensuremath{\predname{q}_0}\xspace}
\newcommand*{\predqloop}{\ensuremath{\predname{q}_r}\xspace} 
\newcommand*{\predright}{\predname{R}}
\newcommand*{\predfuture}{\predname{F}}
\newcommand*{\predend}{\predname{End}}

\newcommand*{\predcopyright}{\predname{C_R}}
\newcommand*{\predcopyleft}{\predname{C_L}}
\newcommand*{\predzero}{\predname{0}}
\newcommand*{\predone}{\predname{1}}
\newcommand*{\predreal}{\predname{Real}}

\newcommand*{\nextBr}{\predname{nextBr}}

\newcommand*{\semterms}{\sfname{semterms}}
\newcommand*{\setbowtie}{\sfname{bowtie}}
\newcommand*{\predfun}{\sfname{pred}}

\newcommand*{\seqstateatoms}{\mathcommand{\mathcal{A}}}

\newcommand*{\configs}{\sfname{configs}}

\setcopyright{cc}
\setcctype{by}
\acmJournal{PACMMOD}
\acmYear{2025} \acmVolume{3} \acmNumber{2 (PODS)}
\acmArticle{109} \acmMonth{5} \acmPrice{}\acmDOI{10.1145/3725246}

\begin{document}

\title{Restricted Chase Termination: You Want More than Fairness}

\author{David Carral}
\affiliation{
  \institution{LIRMM, Inria, University of Montpellier, CNRS}
  \city{Montpellier}
  \country{France}
}
\email{david.carral@inria.fr}
\orcid{https://orcid.org/0000-0001-7287-4709}

\author{Lukas Gerlach}
\affiliation{
  \institution{Knowledge-Based Systems Group, TU Dresden}
  \city{Dresden}
  \country{Germany}
}
\email{lukas.gerlach@tu-dresden.de}
\orcid{https://orcid.org/0000-0003-4566-0224}

\author{Lucas Larroque}
\affiliation{
  \institution{Inria, DI ENS, ENS, CNRS, PSL University}
  \city{Paris}
  \country{France}
}
\email{lucas.larroque@inria.fr}
\orcid{https://orcid.org/0009-0007-2351-2681}

\author{Michaël Thomazo}
\affiliation{
  \institution{Inria, DI ENS, ENS, CNRS, PSL University}
  \city{Paris}
  \country{France}
}
\email{michael.thomazo@inria.fr}
\orcid{https://orcid.org/0000-0002-1437-6389}


\begin{abstract}
  The chase is a fundamental algorithm with ubiquitous uses in database theory. Given a database and a set of existential rules (aka tuple-generating dependencies), it iteratively extends the database to ensure that the rules are satisfied in a most general way. This process may not terminate, and a major problem is to decide whether it does. This problem has been studied for a large number of chase variants, which differ by the conditions under which a rule is applied to extend the database. Surprisingly, the complexity of the universal termination of the restricted (aka standard) chase is not fully understood. We close this gap by placing universal restricted chase termination in the analytical hierarchy. This higher hardness is due to the fairness condition, and we propose an alternative condition to reduce the hardness of universal termination.
\end{abstract}




\received{December 2024}
\received[revised]{February 2025}
\received[accepted]{March 2025} 


\maketitle

\section{Introduction}
\label{section:introduction}


The chase is a fundamental algorithm in database theory that is applied to address a wide range of problems.
For instance, it is used to check containment of queries under constraints, in data exchange settings, or to solve ontology-based query answering; see the introductions of \cite{oblivious-rule-set-termination-undecidable-tr,anatomy-chase} for more information.
Technically speaking, the chase is a bottom-up materialisation procedure that attempts to compute a universal model (a model that can be embedded into all other models via homomorphism) for a knowledge base (KB), consisting of an (existential) rule set\footnote{Other researchers refer to these first-order formulas as ``tuple generating dependencies'' or simply as ``TGDs''.} and a database.
\begin{example}
\label{example:example-bicycle}
Consider the KB $\kb_1 = \langle \aruleset, \db \rangle$ where $\db$ is the database $\{\pn{Bicycle}(b)\}$ and \aruleset contains:
\begin{align*}
\forall x . \pn{Bicycle}(x) &\to \exists y . \pn{HasPart}(x, y) \wedge \pn{Wheel}(y) & \forall x, y . \pn{HasPart}(x, y) &\to \pn{IsPartOf}(y, x) \\
\forall x . \pn{Wheel}(x) &\to \exists y . \pn{IsPartOf}(x, y) \wedge \pn{Bicycle}(y) & \forall x, y . \pn{IsPartOf}(x, y) &\to \pn{HasPart}(y, x)
\end{align*}
Then, $\{\pn{Bicycle}(b), \pn{HasPart}(b, t), \pn{IsPartOf}(t, b), \pn{Wheel}(t)\}$ is a universal model of \kb.
\end{example}


Although there are many variants of the chase, they all implement a similar strategy.
Namely, they start with the database and then, in a step-by-step manner, extend this structure with new atoms to satisfy the rules in the input rule set in a most general way.
Since none of these variants are guaranteed to terminate (some KBs do not even admit finite universal models), it is only natural to wonder about their respective halting problems \cite{DBLP:conf/ijcai/BednarczykFO20,oblivious-rule-set-termination-undecidable,DBLP:conf/kr/CarralLMT22,critical-instance,chase-revisited,anatomy-chase}.
Despite intensive efforts, some results have remained open (until now!).
Specifically, prior research has established tight bounds for all classes of chase terminating KBs and rule sets, except for the following:
\begin{itemize}
\item The class \TKR{\forall} of all KBs that only admit finite restricted chase sequences.
\item  The class \TRR{\forall} containing a rule set \aruleset if $\langle \aruleset, \db \rangle \in \TKR{\forall}$ for every database \db.
\end{itemize}
Our main contribution is to show that both classes are $\Pi_1^1$-complete, a surprising result given that these are significantly harder than the corresponding classes for other chase variants \cite{anatomy-chase}.

{ 
\newcommand*{\xsep}{1.175}
\newcommand*{\xa}{0*\xsep}
\newcommand*{\xb}{1*\xsep}
\newcommand*{\xc}{2.1*\xsep}
\newcommand*{\xd}{3.1*\xsep}
\newcommand*{\xe}{4.1*\xsep}
\newcommand*{\xf}{5.1*\xsep}
\newcommand*{\xg}{6.2*\xsep}
\newcommand*{\xh}{7.2*\xsep}
\newcommand*{\xj}{8.2*\xsep}
\newcommand*{\xk}{9.2*\xsep}
\newcommand*{\xl}{10.2*\xsep}
\newcommand*{\xm}{11.2*\xsep}
\newcommand*{\xn}{12.2*\xsep}

\newcommand*{\ysep}{0.9}
\newcommand*{\ya}{0*\ysep}
\newcommand*{\yb}{1*\ysep}

\begin{figure}
\begin{tikzpicture}
\node[draw, fill=black, circle, inner sep=1pt] (xa) at (\xa, \ya) {};
\node[below] at (xa) {\pn{B}:1};
\node[draw, fill=black, circle, inner sep=1pt] (xb) at (\xb, \yb) {};
\node[above] at (xb) {\pn{W}:2};

\path[->] (xa) edge [bend left] node[fill=white,yshift=2mm] {\pn{HP}:2} (xb);
\path[->] (xb) edge [bend left] node[fill=white,yshift=-2mm] {\pn{IP}:3} (xa);

\node[draw, fill=black, circle, inner sep=1pt] (xc) at (\xc, \ya) {};
\node[below] at (xc) {\pn{B}:1};
\node[draw, fill=black, circle, inner sep=1pt] (xd) at (\xd, \yb) {};
\node[above] at (xd) {\pn{W}:2};
\node[draw, fill=black, circle, inner sep=1pt] (xe) at (\xe, \ya) {};
\node[below] at (xe) {\pn{B}:3};
\node[draw, fill=black, circle, inner sep=1pt] (xf) at (\xf, \yb) {};
\node[above] at (xf) {\pn{W}:4};

\path[->] (xc) edge [bend left] node[fill=white,yshift=2mm] {\pn{HP}:2} (xd);
\path[->] (xd) edge [bend left] node[fill=white,yshift=-2mm] {\pn{IP}:6} (xc);
\path[->] (xd) edge [bend left] node[fill=white,yshift=2mm] {\pn{IP}:3} (xe);
\path[->] (xe) edge [bend left] node[fill=white,yshift=-2mm] {\pn{HP}:7} (xd);
\path[->] (xe) edge [bend left] node[fill=white,yshift=2mm] {\pn{HP}:4} (xf);
\path[->] (xf) edge [bend left] node[fill=white,yshift=-2mm] {\pn{IP}:5} (xe);

\node[draw, fill=black, circle, inner sep=1pt] (xg) at (\xg, \ya) {};
\node[below] at (xg) {\pn{B}:1};
\node[draw, fill=black, circle, inner sep=1pt] (xh) at (\xh, \yb) {};
\node[above] at (xh) {\pn{W}:2};
\node[draw, fill=black, circle, inner sep=1pt] (xj) at (\xj, \ya) {};
\node[below] at (xj) {\pn{B}:3};
\node[draw, fill=black, circle, inner sep=1pt] (xk) at (\xk, \yb) {};
\node[above] at (xk) {\pn{W}:5};
\node[draw, fill=black, circle, inner sep=1pt] (xl) at (\xl, \ya) {};
\node[below] at (xl) {\pn{B}:7};
\node[draw, fill=white, circle, inner sep=1pt] (xm) at (\xm, \yb) {};
\node[above] at (xm) {\pn{W}:9};

\path[->] (xg) edge [bend left] node[fill=white,yshift=2mm] {\pn{HP}:2} (xh);
\path[->] (xh) edge [bend left] node[fill=white,yshift=-2mm] {\pn{IP}:4} (xg);
\path[->] (xh) edge [bend left] node[fill=white,yshift=2mm] {\pn{IP}:3} (xj);
\path[->] (xj) edge [bend left] node[fill=white,yshift=-2mm] {\pn{HP}:6} (xh);
\path[->] (xj) edge [bend left] node[fill=white,yshift=2mm] {\pn{HP}:5} (xk);
\path[->] (xk) edge [bend left] node[fill=white,yshift=-2mm] {\pn{IP}:8} (xj);
\path[->] (xk) edge [bend left] node[fill=white,yshift=2mm] {\pn{IP}:7} (xl);
\path[->] (xl) edge [bend left,dotted] node[fill=white,yshift=-2mm] {\pn{HP}:10} (xk);
\path[->] (xl) edge [bend left,dotted] node[fill=white,yshift=2mm] {\pn{HP}:9} (xm);

\end{tikzpicture}
\caption{Three Different Restricted Chase Sequences for the KB $\kb_1$ from \cref{example:example-bicycle}}
\label{figure:restricted-sequence-example1}
\end{figure}
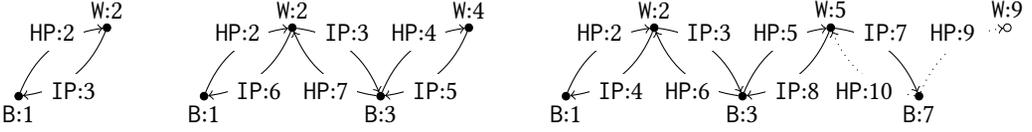
} 

The restricted chase differs from other variants in that it introduces new terms to satisfy existential quantifiers in rules only if these are not already satisfied by existing terms.
Because of this, the order of rule applications impacts the termination of a chase sequence.
For instance, the KB $\kb_1$ from \cref{example:example-bicycle} admits finite and infinite restricted chase sequences; some of these are represented in \cref{figure:restricted-sequence-example1}, 
where atoms are numbered to denote the sequence step at which they were introduced.

\TKR{\forall} has been claimed to be recursively enumerable (\re) in \cite{anatomy-chase}, probably with the following procedure in mind: 
given an input KB, compute all of its restricted chase sequences in parallel, and halt and accept if all of them are finite.
Alas, this strategy does not work as there are terminating input KBs that admit infinitely many finite sequences that are of ever-increasing length.
\begin{example}
\label{example:emergency-brake}
Consider the KB $\kb_2 = \langle \aruleset, \db\rangle$ where \db is the database $\{\pn{Real}(a), \pn{E}(a, c), \pn{E}(c, b), \pn{Real}(c)$, $\pn{E}(b, b), \pn{Brake}(b)\}$ and $\aruleset$ is the rule set that contains all of the following:
\begin{align*}
\forall x, y, z . \pn{Real}(x) \wedge \pn{E}(x, y) \wedge \pn{Real}(y) \wedge \pn{Brake}(z) &\to \exists v . \pn{E}(y, v) \wedge \pn{E}(v, z) \wedge \pn{Real}(v) \\
\forall x . \pn{Brake}(x) &\to \pn{Real}(x)
\end{align*}
For any $k \geq 1$, there is a restricted chase sequence of $\kb_2$ that yields the (finite) universal model $D \cup \{\pn{E}(c, t_1)\} \cup \{\pn{E}(t_i, t_{i+1}) \mid i < k\} \cup \{\pn{E}(t_i, b), \pn{Real}(t_i) \mid i \leq k\} \cup \{\pn{Real}(b)\}$ of $\kb_2$.
Such a sequence is obtained by applying the first rule $k$ consecutive times and then applying the second one once to derive $\pn{Real}(b)$.
After this application, the first rule is satisfied and the restricted chase halts.
\end{example}

The KB $\kb_2$ in the previous example is in \TKR{\forall} because of \emph{fairness}.
This is a built-in condition in the definition of all chase variants that guarantees that the chase yields a model of the KB 
by requiring that, if a rule is applicable at some point during the computation of a sequence, then this rule must be eventually satisfied.
Hence, the second rule in $\kb_2$ must sooner or later be applied in all restricted chase sequences and thus, all such sequences are finite.

The KB in \cref{example:emergency-brake} uses a technique called the \emph{emergency brake}, initially proposed by \citeauthor{DBLP:conf/icdt/KrotzschMR19} in \cite{DBLP:conf/icdt/KrotzschMR19}.
The idea is to connect every term in the chase to a special term (the constant $b$ in this example) that is not ``\pn{Real}'' and acts as a ``\pn{Brake}''.
Eventually, this term becomes ``\pn{Real}'' because of fairness, all existential restrictions are satisfied, and the restricted chase halts.
The emergency brake allows to grow the chase for an arbitrary number of steps whilst guaranteeing its termination.
By activating infinite sequences of emergency brakes, we emulate the eternal recurrence often displayed by $\Pi_1^1$-complete problems and thus define the reductions that lead to our main results.

After presenting the necessary preliminaries in \cref{section:preliminaries}, we discuss related work in \cref{section:related-work}.
Then, we show that \TKR{\forall} and \TRR{\forall} are $\Pi_1^1$-complete in \cref{section:kb-termination,section:rule-set-termination}, respectively.
In \cref{section:fairness}, we propose an alternative to fairness for the restricted chase that simplifies its universal termination problem.
We conclude with a  brief discussion about future work in \cref{section:conclusions}.

\section{Preliminaries}
\label{section:preliminaries}


\paragraph{First-Order Logic}

Consider pairwise disjoint, countably infinite sets of \emph{predicates} \Preds, \emph{variables} \Vars, \emph{constants} \Consts, and \emph{nulls} \Nulls.
Every predicate has an \emph{arity} through $\ar : \Preds \to \nat \cup \{0\}$.
Elements in $\Vars \cup \Consts \cup \Nulls$ are called \emph{terms}. 
An \emph{atom} is an expression of the form $\pn{P}(\vt)$ where $\vt$ a list of terms and $\pn{P}$ is a $\vert \vt \vert$-ary predicate.
A \emph{fact} is a variable-free atom.
An \emph{(existential) rule} \arule is a closed first-order formula of the form 
$\forall \vx, \vy. B[\vx, \vy] \to \exists \vz. H[\vy, \vz]$
where 
$\vx$, $\vy$, and $\vz$ are pairwise disjoint lists of variables; $B$ and $H$ are null-free conjunctions of atoms featuring exactly the variables in $\vx, \vy$ and $\vy,\vz$, respectively;  and $H$ is non-empty.
We write $\body(\arule)$ and $\head(\arule)$ to denote $B$ and $H$, respectively; and refer to the list \vy of variables as the \emph{frontier} of \arule.
We omit universal quantifiers for brevity. 
A \emph{database} is a finite fact set without nulls.
A \emph{knowledge base (KB)} is a pair $\langle \aruleset, \db \rangle$ consisting of a finite rule set \aruleset and a database \db.

\paragraph{The Chase}
A \emph{substitution} \subs is a partial mapping from variables to constants or nulls.
For an (arbitrary) expression $\varphi$, let $\subs(\varphi)$ be the expression that results from $\varphi$ by replacing 
all occurrences of every variable $v$ in $\varphi$ by $\subs(v)$ if the latter is defined.
A \emph{trigger} is a pair $\langle \arule, \subs \rangle$ consisting of a rule \arule and a substitution \subs that is defined exactly on the universally quantified variables in \arule.
The \emph{support} of a trigger $\angles{\arule,\subs}$ is $\trigsupport(\angles{\arule,\subs})=\subs(\body(\arule))$.
A trigger $\langle \arule, \subs \rangle$ 
is \emph{loaded} for a fact set $F$ if this fact set includes
its support; 
and \emph{obsolete} for $F$ if there exists a substitution $\subs'$ that extends \subs to the existential variables in \arule 
such that $\subs'(\head(\arule)) \subseteq F$.
The \emph{output}  of a trigger $\angles{\arule,\subs}$ that is not obsolete for $F$ is $\trigoutput((\angles{\arule,\subs}) =\subs'(\head(\arule))$, where $\subs'$ is some substitution that extends $\subs$ by mapping every existential variable in $\arule$ to a fresh null.
A \emph{$\aruleset$-trigger} is a trigger with a rule in \aruleset.

\begin{definition}
\label{definition:chase}
A \emph{(restricted) chase derivation} for a KB $\langle \aruleset, \db \rangle$ is 
a possibly infinite sequence $F_0, F_1, \dots$ of fact sets such that 
(1) $F_0 = \db$ and,
(2) for each $i \geq 0$, there is some $\aruleset$-trigger $\langle \arule, \subs \rangle$ that is loaded and not obsolete for $F_i$ such that $F_{i+1} = F_i \cup \trigoutput(\angles{\arule,\subs})$.
Such a chase derivation is a \emph{(restricted) chase sequence} if,
(3) for every $\aruleset$-trigger $\lambda$ and every $i \geq 0$ such that $\lambda$ is loaded for $F_i$, there is some $j \geq i$ such that $\lambda$ is obsolete for $F_j$.
\end{definition}

Condition (3) is known as \emph{fairness}.
Note that, if no appropriate trigger according to condition (2) exists for some $i \geq 0$, then the sequence necessarily ends at $F_i$.
The \emph{result} of a chase sequence \chaseseq is the union of all fact sets in \chaseseq.
It is well-known that the result $F$ of any chase sequence for a KB $\KB = \langle \aruleset, \db \rangle$
is a \emph{universal model} for \KB.
That is, every model of \KB can be homomorphically embedded into $F$, which is also a model of this theory.
Note that, if we consider infinite sequences, the result of the chase may not be a model of \KB if we disregard fairness.

A chase sequence \emph{terminates} if it is finite. 
A KB \emph{existentially terminates} if it admits a terminating chase sequence; it \emph{universally terminates} if all of its chase sequences terminate.
A rule set \aruleset \emph{existentially terminates} if every KB with \aruleset existentially terminates; it \emph{universally terminates} if every KB with \aruleset universally terminates.
The classes of knowledge bases that existentially and universally terminate 
are denoted by $\TKR{\exists}$ and $\TKR{\forall}$, respectively.
The classes of rule sets that existentially and universally terminate
are denoted by $\TRR{\exists}$ and $\TRR{\forall}$, respectively.
We also consider similar classes for the oblivious and core chase variants, which we denoted in the obvious manner.
For instance, \TRO{\exists} is the set of all rule sets that existentially terminate for the oblivious chase.

\paragraph{Turing Machines}
As per our definition, all machines reuse the same \emph{initial} state.
Moreover, machines do not write blanks and cannot access accepting or rejecting states; these are not relevant in our context because we only consider halting problems.

\begin{definition}
\label{definition:machine}
  A \emph{(non-deterministic Turing) machine} is a tuple $\langle Q, \Gamma, \delta \rangle$ 
  where $Q$ is a set of states that contains the \emph{initial state} \qinit, 
  $\Gamma$ is a tape alphabet with $\Gamma \supseteq \{0, 1\}$ and $\blank \notin \Gamma$, 
  and $\delta$ is a transition function for $Q$.
  That is, $\delta$ is a function that maps from $Q \times \Gamma \cup \{\blank\}$ to $\mathcal{P}(Q \times \Gamma \times \{\goleft, \goright\})$.
\end{definition}

\begin{definition}
A \emph{configuration} for a machine $\langle Q, \Gamma, \delta \rangle$ is a tuple $\langle n, t, p, q \rangle$ where $n$ is a natural number; $t : \{1, \ldots, n\} \to \Gamma \cup \{\blank\}$ is a function such that $t(n) = \blank$, and $t(i+1) = \blank$ if $t(i) = \blank$ for some $1 \leq i < n$;
$p$ is a number in $\{1, \ldots, n\}$; and $q$ is a state in $Q$.
The \emph{starting configuration} on some \emph{word} $w_1, \ldots, w_n \in \{0, 1\}^*$ is the tuple $\langle n+1, t, 1, \qinit \rangle$ where $t$ is the function that maps $1$ to $w_1$, $2$ to $w_2$, \ldots, $n$ to $w_n$, and $n+1$ to  $\blank$.
\end{definition}

For a configuration $\langle n, t, p, q \rangle$, we use $t$ to encode the contents of the tape at each position; moreover, we use $p$ and $q$ to encode the position of the head and the state of the machine, respectively.
Note that elements of the tape alphabet $\Gamma$ may not occur after a blank symbol in such a configuration.


\begin{definition}
\label{definition:machine-configuration}
Consider a machine $M = \langle Q, \Gamma, \delta \rangle$ and a configuration $\rho = \langle n, t, p, q \rangle$ with $q \in Q$.
Then, let $\textsf{Next}_M(\rho)$ be the smallest set that, for every $\langle r, a, \mathop{\leftrightarrow} \rangle \in \delta(t(p), q)$ with $\mathop{\leftrightarrow} = \goright$ or $p \geq 2$, contains the configuration $\langle n + 1, t', p', r \rangle$ where:
\begin{itemize}
\item Let $t'(p) = a$, let $t'(n+1) = \blank$, and let $t'(i) = t(i)$ for every $1 \leq i \leq n$ with $i \neq p$.
\item If $\mathop{\leftrightarrow} = \goleft$, then $p' = p-1$; otherwise, $p' = p+1$.
\end{itemize}
\end{definition}


As described above, any given machine defines a function that maps configurations to sets of configurations.
An exhaustive traversal through a path in this possibly infinite tree of configurations that begins with a starting configuration yields a run:

\begin{definition}
\label{definition:machine-run-halting}
A \emph{run} of a machine $M$ on a configuration $\rho_1$ is a possibly infinite sequence $S = \rho_1, \rho_2, \ldots$ of configurations such that $\rho_{i+1}$ is in $\textsf{Next}_M(\rho_i)$ for every $1 \leq i < \vert S \vert$, and $\textsf{Next}_M(\rho_{\vert S \vert}) = \emptyset$ if $S$ is finite.
A \emph{partial run} of $M$ on $\rho_1$ is a sequence of configurations that can be extended into a run of $M$ on $\rho_1$.
A \emph{(partial) run} of $M$ on a word $\vec{w}$ is a (partial) run on the starting configuration of $\vec{w}$.
\end{definition}

\paragraph{Computability Theory}

The \emph{arithmetical hierarchy} consists of classes of formal languages $\Sigma_i^0$ with $i \geq 1$ where $\Sigma_1^0$ is the class of all semi-decidable languages 
and $\Sigma_{i+1}^0$ is obtained from $\Sigma_i^0$ with a Turing jump \cite{RogersArithmeticalHierarchy}.
The co-classes are denoted by $\Pi_i^0$.
Equivalently, these classes can be viewed as the sets of natural numbers definable by first-order logic formulas with bounded quantifier alternation.
That is, $\Sigma_i^0$ is the class of sets of natural numbers definable with a formula of the form $\exists \vec{x}_1 \forall \vec{x}_2 \dots Q_i \vec{x}_i. \phi[x, \vec{x}_1, \dots, \vec{x}_i]$ 
where $\phi$ is a quantifier-free formula and $Q_i$ is $\exists$ if $i$ is odd or $\forall$ otherwise.
For $\Pi_i^0$, the alternation starts with $\forall$.
We also the first level of the \emph{analytical hierarchy}; that is, $\Sigma_1^1$ and $\Pi_1^1$
\cite{RogersArithmeticalHierarchy}.
The analytical hierarchy can analogously be defined using second-order formulae with bounded second-order quantifier alternation.
In the following, we introduce complete problems for these classes that we later use in our reductions.
Consider a machine $M$ and a state \qloop.
\begin{itemize}
\item The machine $M$ is \emph{non-recurring through \qloop} on some word $\vec{w}$ if every run of $M$ on $\vec{w}$ features \qloop finitely many times.
\item It is \emph{universally non-recurring through \qloop} if it is non-recurring through \qloop on all words.
\item It is \emph{robust non-recurring through \qloop} if every run of $M$ on any configuration features \qloop finitely many times.
\end{itemize}
We obtain $\Pi_1^1$-completeness of the first problem by adjusting a proof from the literature \cite{pi_one_one_problems} and for the latter two using simple reductions that we define in \cref{section:tm-reductions}.

\begin{toappendix}
\section{$\Pi_1^1$-complete Turing Machine Problems for Reductions}
\label{section:tm-reductions}

Our definition of non-recurring machines differs slightly from descriptions found in previous literature.
Indeed, Harel showed that the following problem is $\Pi_1^1$-complete: decide if a (non-deterministic Turing) machine admits a run on the empty word that features the initial state \qinit infinitely many times (see Corollary~6.2 in \cite{pi_one_one_problems}).
Our definition is slightly different since we choose a different state $\qloop$ to keep track of this infinite recurrence; note that this state may be different from the initial state.
Fortunately, the choice of the initial state in the proof of Corollary 6.2 of Harel \cite{pi_one_one_problems} is arbitrary, making it straightforward to adapt his proof to any given state.
We first prove this in Section~\ref{subsec:tm-reductions:empty_word}, and then use this result to get $\Pi_1^1$-completeness for the other Turing machine problems we consider in Section~\ref{subsec:tm-reductions:univ_and_robust}.

\subsection{Non-Recurrence on the Empty Word}\label{subsec:tm-reductions:empty_word}

To show that checking if a machine is non-recurring on the empty word is $\Pi_1^1$-complete, we adapt the proof of Corollary 6.2 in \cite{pi_one_one_problems}.
To do so, we first need to introduce some preliminary notions.
A \emph{list} is a finite sequence.
The \emph{concatenation} of two lists $u = u_1,\dots, u_n$ and $v = v_1,\dots, v_m$ is the list $u \cdot v = u_1,\dots, u_n, v_1,\dots, v_m$.
A list $u_1, \ldots, u_n$ with $n \geq 2$ is the \emph{child} of $u_1, \ldots, u_{n-1}$.
A list $u$ is an \emph{ancestor} of another list $v$, written $u \prec v$, if $u$ is a prefix of $v$; that is, if $u \cdot w = v$ for some list $w$.

\begin{definition}
\label{definition:omega-trees}
An \emph{$\omega$-tree} $T$ is a set of lists of natural numbers closed under $\prec$.
A \emph{node} is an element in $T$; a \emph{leaf} is a node without children in $T$.
Such a tree is \emph{computable} if so is the following function:
    \[\chi_T(u)=\begin{cases}
        0 & \text{if } u\notin T\\
        1 & \text{if } u\in T\text{ and } u \text{ is a leaf}\\
        2 & \text{if } u\in T\text{ and } u \text{ is not a leaf}
    \end{cases}\]
A possibly infinite sequence of natural numbers is a \emph{branch} of $T$ if the latter contains every finite prefix of the former.
Such a tree is \emph{well founded} if all of its branches are finite.
\end{definition}

In the following, we identify a computable $\omega$-tree $T$ with the machine that computes the function $\chi_T$.
Note that this is a machine that implements a function mapping lists of natural numbers to elements of $\{0,1,2\}$ as indicated in \cref{definition:omega-trees}.
Checking if such a machine does correspond to a well-founded tree is a $\Pi_1^1$-complete problem.

\begin{lemma}[\cite{RogersArithmeticalHierarchy}, Theorem 16]\label{lem:wf_omega_trees_are_pi_one_one}
Checking if a computable $\omega$-tree is well founded is $\Pi_1^1$-complete.
\end{lemma}

\begin{definition}
For a natural number $k \geq 0$, a \emph{$k$-tree} $T$ is an $\omega$-tree that does not contain sequences with numbers larger than $k$.
A \emph{b-tree} (b for bounded) is a $k$-tree for some $k \geq 0$.
A \emph{marked b-tree} is a pair $(T, \mu)$ consisting of a b-tree $T$ and a \emph{marking function} $\mu$; that is, a function from $T$ to $\set{0,1}$.
A marked b-tree is \emph{computable} if the following function is computable:
\[\chi_T^\mu(u)=\begin{cases}
 0 & \text{if } u\notin T\\
1 & \text{if } u\in T\text{ and } u \text{ is marked }(\textit{that is, } \mu(u)=1)\\
2 & \text{if } u\in T\text{ and } u \text{ is not marked}
\end{cases}\]
A marked b-tree is \emph{recurring} if it has a branch with infinitely many marked prefixes.
\end{definition}

As we do for computable $\omega$-trees, we identify a computable marked b-tree $(T, \mu)$ with the decider that implements the function $\chi_T^\mu$.

\begin{lemma}[\cite{pi_one_one_problems}, Corollary 5.3]\label{lem:rec_b_trees_are_pi_one_one}
Checking if a computable b-tree is non-recurring is $\Pi_1^1$-complete.
\end{lemma}

We are ready now to show the main result in this subsection.

\begin{proposition}\label{prop:rec_halt_pi_one_one}
The problem of checking if a machine is non-recurring through some state \qloop on the empty word $\varepsilon$ is $\Pi_1^1$-complete.
\end{proposition}
\begin{proof}
To show membership, we present a reduction that maps a machine $M = (Q,\Gamma,\delta)$ to a computable marked b-tree $(T, \mu)$ such that $M$ is non-recurring through a given state $\qloop \in Q$ on the empty word $\varepsilon$ if and only if $(T, \mu)$ is non-recurring.
To define $(T, \mu)$, we consider an (arbitrarily chosen) enumeration $q_1, \dots, q_n$ of the states in $Q$.
\begin{itemize}
\item Let $T$ be the set containing a list of natural numbers $i_1,\dots, i_n$ if there is a partial run $\rho_1, \ldots, \rho_n$ of $M$ on $\varepsilon$ such that $\rho_j$ features the state $q_{i_j}$ for every $1 \leq j \leq n$.
\item Let $\mu$ be the function that maps a list $u \in T$ to $1$ if and only if $q_i = \qloop$ where $i$ is the last element in $u$.
That is, if the last element of $u$ is the index that corresponds to $\qloop$ in the enumeration $q_1, \dots, q_n$.
\end{itemize}
For every infinite branch of $T$, there is an infinite run of $M$ and vice-versa.
Furthermore, by the definition of $\mu$, a branch of $(T,\mu)$ containing infinitely many marked nodes corresponds to a run of $M$ visiting $\qloop$ infinitely many times.
Therefore, $M$ is non-recurring through \qloop if and only if $(T,\mu)$ is non-recurring.
    
For hardness, we present a reduction that maps a computable $\omega$-tree $T$ to a non-deterministic machine $M = (Q,\Gamma,\delta)$ such that $T$ is well-founded if and only if $M$ is non-recurring through a state $\qloop \in Q$ on the empty word $\varepsilon$.
Intuitively, the machine $M$ proceeds by doing a traversal of the full $\omega$-tree; formally, it implements the following instructions on input $\varepsilon$:
\begin{enumerate}
\item Initialise the variable $u = 0$, which stores a list of natural numbers.
\item If $u \notin T$, replace the last element $i$ in $u$ with $i+1$.
\item If $u \in T$, make a non-deterministic choice between the following options:
\begin{enumerate}
\item Replace the last element $i$ in $u$ with $i+1$.
\item Append $0$ to the list stored in $u$ and visit the state \qloop.
\end{enumerate}
\item Go to (2).
\end{enumerate}
We can effectively check if a list $u$ is a node in $T$ above because $T$ is a computable $\omega$-tree and hence, so is function $\chi_T$.
Intuitively, each run of $M$ on the empty word corresponds to a traversal of a branch in $T$; note how we use non-determinism in (3) to alternatively visit the sibling (Instruction~3.a) or the child (Instruction~3.b) of a node in the tree.
Furthermore, note that $M$ only visits \qloop when it moves deeper on a given branch; that is, when it executes instruction (3.b).
Therefore, there is a run of $M$ visiting $\qloop$ infinitely often if and only if there is an infinite branch in $T$.
\end{proof}

\subsection{Reductions between Turing Machine Problems}
\label{subsec:tm-reductions:univ_and_robust}

\begin{proposition}\label{prop:univ_rec_halt_pi_one_one}
The problem of checking if a machine is universally non-recurring through a given state \qloop is $\Pi_1^1$-complete.
\end{proposition}
\begin{proof}
To show membership, we present a reduction that maps a machine $M$ to another machine $M'$ such that $M$ is universally non-recurring through a state $\qloop$ if and only if $M'$ is non-recurring through a state $q_r'$ on $\varepsilon$.
On input $\varepsilon$, the machine $M'$ first guesses some input word and then simulates $M$ on this input.
Formally, it executes the following instructions:
\begin{enumerate}
\item Make a non-deterministic choice to decide whether to go to (2) or to (3).
\item Replace the first occurrence of the blank symbol \blank in the input tape with some non-deterministically chosen symbol in the input alphabet of $M$. Then, go to (1).
\item Simulate $M$ on the (finite) word written down in the input tape.
During this simulation, visit $q_r'$ whenever $M$ would have visited $\qloop$.
\end{enumerate}
Note that there are infinite runs of $M'$ on $\varepsilon$ where the machine never executes Instruction~3.
This does not invalidate our reduction since $M'$ never visits $q_r'$ in these branches.

To show hardness, we present a reduction that maps a machine $M$ to another machine $M'$ such that $M$ is non-recurring through a state $\qloop$ on $\varepsilon$ if and only if $M'$ is universally non-recurring through a state $q_r'$.
The machine $M'$ first discards its input by replacing it with a special symbol that is treated like the blank symbol \blank.\footnote{We consider a special symbol here because, as per our definition, machines may not print the blank symbol \blank.}
Then, $M'$ simulates $M$ on $\varepsilon$; during this simulation, $M'$ visits $q_r'$ whenever $M$ would have visited $\qloop$.
\end{proof}

\begin{proposition}\label{prop:rec_mort_pi_one_one}
Checking if a machine is robust non-recurring through \qloop is $\Pi_1^1$-complete.
\end{proposition}
\begin{proof}
To show membership, we present a reduction from a machine $M$ to a machine $M'$ such that $M$ is robust non-recurring through a state \qloop if and only if $M'$ is universally non-recurring through a state $q_r'$.
The machine $M'$ scans its input and halts if it does not encode a configuration of $M$.
Otherwise, $M'$ simulates $M$ starting on this input configuration; during this simulation, $M'$ visits $q_r'$ whenever $M$ would have visited $\qloop$.

To show hardness, we present a reduction from a machine $M$ to another machine $M'$ such that $M$ is non-recurring through a state $q_r$ on the empty word $\varepsilon$ if and only if $M'$ is robust non-recurring through a state $q_r'$.
The machine $M'$ executes the following instructions:
\begin{enumerate}
\item Halt if the input does not contain some configuration $\rho$ of $M$.
\item If the configuration in the tape $\rho$ features the special state \qloop, then visit $q_r'$.
\item After the encoding of $\rho$ in the tape,
(non-deterministically) simulate a run of $M$ on $\varepsilon$ until it terminates or you reach the configuration $\rho$.
If the run terminates, without finding $\rho$, halt. Otherwise, continue in (4).
\item If $\textsf{Next}_{M}(\rho)$ is empty, halt.
Otherwise, replace the configuration $\rho$ in the tape with a non-deterministically chosen configuration in $\textsf{Next}_{M}(\rho)$, and go to (1).
\end{enumerate}
Intuitively speaking, Instruction~3 implements a reachability check for the configuration $\rho$ in the tape.
That is, this procedure ensures that this configuration is reachable from the starting configuration of $M$ on the empty word $\varepsilon$ by some run. Note that the reachability check makes non-deterministic choices itself. So it can happen that $M'$ terminates early or even runs forever because it picks the wrong run in Instruction~3. This does not invalidate our reduction though since on those runs, $M'$ only visits $q_r'$ finitely many times. 

If $M'$ is robust non-recurring through $q_r'$, then it is also non-recurring on the starting configuration with the encoding of the starting configuration of $M$ on $\varepsilon$.
Since $M'$ uses non-determinism to simulate all possible runs $M$ on $\varepsilon$ and visits $q_r'$ whenever $M$ would have visited $q_r$, we conclude that $M$ is non-recurring through $q_r$ on $\varepsilon$.


Suppose that there is a configuration $\rho'$ of $M'$ that may lead to a run that visits $q_r'$ infinitely many times.
In turn, this implies that there is a configuration $\rho$ of $M$ that leads to a run of $M$ that visits $q_r$ infinitely many times.
Moreover, all the configurations of $M$ in this infinite run are reachable from the start configuration of $M$ on $\varepsilon$ because of the check implemented in Instruction~3.
Therefore, $M$ is recurring through $q_r$ on the empty word.
\end{proof}
\end{toappendix}

\section{Related Work}
\label{section:related-work}

\paragraph{Novel Notation}
The notation introduced in \cref{section:preliminaries} to refer to classes of terminating KBs and rule sets differs from previous literature \cite{anatomy-chase}; for instance, we write  $\TRR{\forall}$ instead of \CTaa.
Moreover, given some database \db, we do not consider a class such as \CTda \cite{anatomy-chase}, which contains a rule set $\aruleset$ if $\langle \aruleset, \db \rangle$ universally terminates for the restricted chase.
For our purposes, it is clearer to consider a single class of terminating KBs (such as $\TKR{\forall}$) instead of one class of terminating rule sets for every possible database because of the following result.

\begin{proposition}
   For a database $\db'$, a quantifier $Q \in \{\forall, \exists\}$, and a chase variant $\textit{var} \in \{\textit{obl}, \textit{rest}, \textit{core}\}$; 
  there is a many-one reduction from $\TerminatingKBs{Q}{var}$ to $\CT^\textit{var}_{D'Q}$ and vice-versa.
\end{proposition}
\begin{proof}
  There is a many-one reduction $\CT^\textit{var}_{D'Q}$ to \TerminatingKBs{Q}{var} since, for a rule set \aruleset, we have that $\aruleset \in \CT^\textit{var}_{D'Q}$ if and only if $\langle \aruleset, \db' \rangle \in \TerminatingKBs{Q}{var}$.
To show that there is a many-one reduction in the other direction we describe a computable function that maps a KB $\kb = \langle \aruleset, \db \rangle$ into the rule set $\aruleset'$ such that $\kb \in \TerminatingKBs{Q}{var}$ if and only if $\aruleset' \in \CT^\textit{var}_{D'Q}$.
Namely, let $\aruleset'$ be the rule set that results from applying the following modifications to $\aruleset$: (i) replace all occurrences of every predicate $P$ with a fresh predicate $P'$, (ii) add the conjunction $\bigwedge_{P(\vec{c}) \in \db} P'(\vec{c})$ to the body of every rule, and (iii) add the rule $\to \bigwedge_{P(\vec{c}) \in \db} P'(\vec{c})$.
The reduction is correct because one can easily establish a one-to-one correspondence between the sequences of $\kb$ and those of $\langle \aruleset', \db' \rangle$ once we ignore the single trigger with $\to \bigwedge_{P(\vec{c}) \in \db} P'(\vec{c})$ at the beginning of every sequence of the latter KB.
Note that the sets of facts produced at subsequent steps of these corresponding sequences are identical modulo replacement of all occurrences of every predicate $P$ by $P'$.
\end{proof}

\renewcommand{\arraystretch}{1.1}
\begin{table}
\centering
\begin{tabular}{ | >{\centering\arraybackslash}p{6em} | >{\centering\arraybackslash}p{7em} | >{\centering\arraybackslash}p{7em} | >{\centering\arraybackslash}p{7em} | >{\centering\arraybackslash}p{7em} |}
\cline{2-5}
\multicolumn{1}{c|}{} & \multicolumn{2}{c|}{\textbf{KB}} & \multicolumn{2}{c|}{\textbf{Rule Set}} \\ \cline{2-5}
\multicolumn{1}{c|}{} & \textbf{Sometimes}	 & \textbf{Always} & \textbf{Sometimes} & \textbf{Always} \\ \hline
\textbf{Oblivious} & \multicolumn{2}{c|}{\re-complete \cite{chase-revisited}}  & \multicolumn{2}{c|}{\re-complete \cite{oblivious-rule-set-termination-undecidable,critical-instance}} \\ \hline
\textbf{Restricted} & \re-complete \cite{chase-revisited} & $\Pi^1_1$-complete & $\Pi_2^0$-complete \cite{anatomy-chase} & $\Pi_1^1$-complete \\ \hline
\textbf{Core} & \multicolumn{2}{c|}{\re-complete \cite{chase-revisited}} & \multicolumn{2}{c|}{$\Pi_2^0$-complete \cite{anatomy-chase}} \\ \hline
\end{tabular}
\caption{Undecidability status of the main decision problems related to chase termination; the results presented without citations refer to the main contributions of this article}
\label{table:checking-chase-termination}
\end{table}

\paragraph{Chase Termination in the General Case}
All decision problems related to chase termination are undecidable.
However, these are complete for different classes within the arithmetical and analytical hierarchies, as summarised in \cref{table:checking-chase-termination}.
In the following paragraphs, we discuss some simple proofs as well as the relevant references to understand all of the results in this table.

One can readily show via induction that, if a fact occurs in some oblivious chase sequence of some KB, then it also occurs in all oblivious chase sequences of this KB.
Hence, all such chase sequences of a KB yield the same result, and thus we conclude that $\TKO{\exists} = \TKO{\forall}$ and $\TRO{\exists} = \TRO{\forall}$.

\citeauthor{chase-revisited} proved that, if a KB admits a finite universal model, then all of its core chase sequences yield precisely this model and thus all of these sequences are finite; see Theorem~7 in \cite{chase-revisited}.
Regardless of the variant, all terminating chase sequences yield a (not necessarily minimal) finite universal model; hence, if a KB does not admit a finite universal model, then it does not admit any finite chase sequence.
Therefore, we have that either all core chase sequences of a KB are finite or all of them are infinite.
Because of this, we conclude that $\TKC{\exists} = \TKC{\forall}$ and $\TRC{\exists} = \TRC{\forall}$.

To understand why \TKO{\exists} (resp. \TKR{\exists} or \TKC{\exists}) is recursively enumerable (\re), consider the following procedure: given some input KB, compute all of its oblivious (resp. restricted or core) chase sequences in parallel and accept as soon as you find a finite one.
\citeauthor{chase-revisited} proved that \TKR{\exists} is \re-hard.
More precisely, they defined a reduction that takes a machine $M$ as input and produces a KB \kb as output such that $M$ halts on the empty word if and only \kb is in $\TKR{\exists}$; see Theorem~1 in \cite{chase-revisited}.
This reduction works because all restricted chase sequences of \kb yield the same result, which encodes the computation of $M$ on the empty word with a grid-like structure (as we ourselves do in later sections).
One can use the same reduction to show that \TKO{\exists} is also \re-hard.

\citeauthor{chase-revisited} also proved that \TKC{\exists} is \re-hard.
More precisely, they showed that checking if a KB admits a universal model is undecidable; see Theorem~6 in \cite{chase-revisited}.
Moreover, they proved that the core chase is a procedure that halts and yields a finite universal model for an input KB if this theory admits one; see Theorem~7 of the same paper.
Therefore, the core chase can be applied as a semi-decision procedure for checking if a KB admits a finite universal model.

In \cref{section:kb-termination}, we argue that \TKR{\forall} is $\Pi_1^1$-complete.
This contradicts Theorem~5.1 in \cite{anatomy-chase}, which states that  \TKR{\forall} is \re-complete. Specifically, it is claimed that this theorem follows from results in \cite{chase-revisited}, but the authors of that paper only demonstrate that \TKR{\forall} is undecidable without proving that it is in \re.
Before our completeness result, the tightest lower bound was proven by \citeauthor{DBLP:conf/kr/CarralLMT22}, who proved that this class is $\Pi^0_2$-hard; see Proposition~42 in \cite{DBLP:conf/kr/CarralLMT22}.

\citeauthor{critical-instance} proved that \TRO{\exists} is in \re.
More precisely, he showed that a rule set $\aruleset$ is in \TRO{\exists} if and only if the KB $\langle \aruleset, \db_\aruleset^\star\rangle$ is in \TKO{\exists} where $\db_\aruleset^\star = \{\pn{P}(\star, \ldots, \star) \mid \pn{P} \in \Preds(\aruleset)\}$ is the \emph{critical instance} and $\star$ is a special fresh constant; see Theorem~2 in \cite{critical-instance}.
This result follows because one can show that, for any database \db, the (only) result of the oblivious chase of $\langle \aruleset, \db_\aruleset^\star \rangle$ includes the (only) result of the oblivious chase of $\langle \aruleset, \db \rangle$ if we replace all syntactic occurrences of constants in the latter with $\star$.
Since \TKO{\exists} is in \re, we conclude that $\TRO{\exists}$ is also in this class.

\citeauthor{oblivious-rule-set-termination-undecidable} proved that \TRO{\exists} is \re-hard.
More precisely, they presented a reduction that takes a 3-counter machine $M$ as input and produces a rule set $\aruleset$ such that $M$ halts on $\varepsilon$ if and only if $\langle \aruleset, \db_\aruleset^\star\rangle$ is in $\TKO{\exists}$; see Lemma~6 in \cite{oblivious-rule-set-termination-undecidable}.\footnote{We do not think that it is possible to intuitively explain this reduction in a couple of lines. Go read this paper, it's worth it!}
Hence, $M$ halts on the $\varepsilon$ and only if $\aruleset$ is in $\TRO{\exists}$ by Theorem~2 in \cite{critical-instance}.
Furthermore, \citeauthor{DBLP:conf/ijcai/BednarczykFO20} showed that this hardness result holds even when we consider single-head binary rule sets; see Theorem~1.1 in \cite{DBLP:conf/ijcai/BednarczykFO20}.

To understand why \TRR{\exists} is in $\Pi_2^0$, consider the following semi-decision procedure that can access an oracle that decides the \re-complete class $\TKR{\exists}$:
given some input rule set $\aruleset$; iterate through every database $\db$, use the oracle to decide if $\langle \aruleset, \db\rangle$ is in $\TKR{\exists}$, and accept if this is not the case.
Consider an analogous procedure to understand why \TRC{\exists} is in $\Pi_2^0$.

\citeauthor{anatomy-chase} proved that \TRR{\exists} is $\Pi_2^0$-hard.
To show this, they defined two reductions that take a word rewriting system $R$ and a word $\vec{w}$ as input, and produce a rule set $\aruleset_R$ and a database $\db_{\vec{w}}$, respectively.
Then, they proved that $R$ terminates on $\vec{w}$ if and only if the KB $\langle \aruleset_R, \db_{\vec{w}} \rangle$ is in $\TKR{\exists}$; this claim holds because $\langle \aruleset_R, \db_{\vec{w}} \rangle$ only admits a single restricted chase result, which encodes all branches of computation of $R$ on $\vec{w}$ in an implicit tree-like structure.
Therefore, $R$ is uniformly terminating if $\aruleset_R$ is in $\TRR{\exists}$.
To ensure that $\aruleset_R$ is in $\TRR{\exists}$ if $R$ is uniformly terminating, \citeauthor{anatomy-chase} make use of ``flooding'', a technique used in earlier work dealing with datalog boundedness \cite{DBLP:journals/jacm/GaifmanMSV93}.
For a comprehensive presentation of this technique and its applications, see Section~2 of \cite{oblivious-rule-set-termination-undecidable-tr}.
Using the very same reduction, \citeauthor{anatomy-chase} also proved that \TRC{\exists} is $\Pi_2^0$-hard.

In Section~\ref{section:rule-set-termination}, we show that \TRR{\forall} is $\Pi_1^1$-complete.
This contradicts Theorem~5.16 in \cite{anatomy-chase}, where it is stated that this class is $\Pi_2^0$-complete.
The error in the upper-bound of this theorem arose from the assumption that \TKR{\forall} is in \re, which, as previously discussed, is not the case.
Regarding the lower bound, they consider an extended version of this class of rule sets where they allow the inclusion of a single ``denial constraint''; that is, an implication with an empty head that halts the chase if the body is satisfied during the computation of a chase sequence.
They prove that the always restricted halting problem for rule sets is $\Pi_2^0$-hard if one such constraint is allowed.
Our results imply that we do not need to consider such an extension to obtain a higher lower bound.

\paragraph{Chase Termination of Syntactic Fragments}
Undeterred by the undecidability results discussed above, researchers have proven we can decide chase termination if we consider syntactic fragments of existential rules for which query entailment is decidable \cite{DBLP:conf/pods/CalauttiGP15,DBLP:journals/mst/CalauttiP21,DBLP:journals/siamcomp/GogaczMP23,DBLP:conf/icdt/LeclereMTU19}.
Another way of checking termination in practice is to develop acyclicity and cyclicity notions; that is, sufficient conditions for termination and non-termination of the chase.
Indeed, experiments show that we can determine chase termination for a large proportion of real-world rule sets with these checks \cite{DBLP:conf/aaai/GerlachC23,DBLP:conf/kr/0002C23,DBLP:conf/ijcai/CarralDK17,DBLP:journals/jair/GrauHKKMMW13}.

\section{Knowledge base termination}
\label{section:kb-termination}

\begin{theorem}
\label{theorem:kb-completeness}
  The class $\TKR{\forall}$ is $\Pi_1^1$-complete.
\end{theorem}

The theorem immediately follows from the upcoming \cref{lem:kb-termination-mem} and \cref{lem:kb-termination-hard}.

For the membership part, we define a non-deterministic Turing machine that loops on \qloop if and only if there is a non-terminating chase sequence for a given rule set. 

\begin{definition}\label{def:termination-aux-machine}
  Consider a rule set \aruleset. 
  For a fact set $F$, let $\textit{active}(F)$ be the set of all triggers with a rule in \aruleset that are loaded and not obsolete for $F$.
  Let $\mathcal{M}_\aruleset$ be a non-deterministic Turing machine with start state \qinit
  and a designated state \qloop
  that executes the following procedure.

  \begin{enumerate}
    \item Check if the input tape contains a valid encoding of a database. If not, halt.
    \item Initialize two counters $i = j = 0$ and a set of facts $F_0$ containing exactly the encoded database.
    \item If $\textit{active}(F_i)$ is empty, halt. \label{procedure:loop}
    \item Non-deterministically pick a trigger $\langle \arule, \subs \rangle$ from $\textit{active}(F_i)$
      and let $F_{i+1} = F_i \cup \subs'(\head(\arule))$ where $\subs'$ extends $\subs$ by mapping existential variables in $\arule$ to fresh nulls (not occurring in $F_i$).
    \item If all triggers in $\textit{active}(F_j)$ are obsolete for $F_i$, then increment $j$ and visit \qloop once. 
    \item Increment $i$ and go to \ref{procedure:loop}.
  \end{enumerate}
\end{definition}

\begin{lemma}\label{lem:termination-aux-machine}
  For every database \db and rule set \aruleset, 
  there is a run of $\mathcal{M}_\aruleset$ on the encoding of \db that visits \qloop infinitely often 
  if and only if there is a non-terminating chase sequence for $\langle \aruleset, \db \rangle$.
\end{lemma}
\begin{proof}
  Assume that there is a run of $\mathcal{M}_\aruleset$ on the encoding of \db that visits \qloop infinitely many times.
  Then, the sequence $F_0, F_1, \dots$ constructed by $\mathcal{M}_\aruleset$ is an infinite restricted chase derivation for $\langle \aruleset, \db \rangle$ by construction.
  Since \qloop is visited infinitely many times, $j$ grows towards infinity. Therefore, every trigger that is loaded for some $F_j$ with $j \geq 0$ is obsolete for some $i \geq j$; which is exactly \emph{fairness}.
  Hence, the infinite derivation is a proper chase sequence.

  Assume that there is an infinite chase sequence $F_0, F_1, \dots$ for $\langle \aruleset, \db \rangle$.
  By definition, for each $i \geq 0$, there is a trigger $\lambda \in \textit{active}(F_{i})$ that yields $F_{i+1}$.
  Hence, there is a run of $\mathcal{M}_\aruleset$ that non-deterministically picks these triggers.
  Because of fairness, for every trigger $\lambda$ in $\textit{active}(F_j)$ with $j \geq 0$, there is $i \geq j$ such that $\lambda$ is obsolete for $F_i$. 
  Hence, the run of $\mathcal{M}_\aruleset$ visits \qloop infinitely often.
\end{proof}

\begin{lemma}\label{lem:kb-termination-mem}
	Deciding membership in $\TKR{\forall}$ is in $\Pi_1^1$.
\end{lemma}
\begin{proof}
  We show a reduction to non-recurrence through \qloop on the empty word.
  For a given rule set \aruleset, let $\mathcal{M}^D_\aruleset$ be a non-deterministic Turing machine 
  that results from $\mathcal{M}_\aruleset$ by adding an initial step that replaces the initial tape content by an encoding of \db. 
  Then, by \cref{lem:termination-aux-machine}, $\aruleset$ is in $\TKR{\forall}$ if and only if no run of $\mathcal{M}^D_\aruleset$ on the empty input visits \qloop infinitely many times.
\end{proof}

\begin{lemma}\label{lem:kb-termination-hard}
	The class $\TKR{\forall}$ is $\Pi_1^1$-hard.
\end{lemma}

To prove hardness, we reduce non-recurrence through \qloop on the empty word to knowledge base termination. In other words, to a Turing machine $M$, we will associate a database $D_\varepsilon$ and a rule set $\aruleset_M$ such that there exists a run of $M$ on the empty word reaching \qloop infinitely often if and only if the restricted chase of $\aruleset_M$ on $D_\varepsilon$ does not halt. 

A perhaps surprising feature of this reduction is that the restricted chase must halt for rule sets generated from Turing machines that do not halt on the empty word, as long as they reach \qloop only finitely often. As we cannot get any computable bound on the number of steps required to reach \qloop, we must simulate any finite run of the Turing machine in a terminating way. This calls for the use of \emph{emergency brakes} as presented in the introduction. We ``stack'' such brakes, each one being responsible to prevent the non-termination for runs that do not go through \qloop. 

\paragraph{Schema} We will make use of the following predicates. Note that the last position 
usually holds an emergency brake.
We introduce: For each letter $a$ in the Turing machine alphabet or equal to the blank \blank, a binary predicate $\preda$.
For each state $q$ of the Turing machine, a binary predicate $\predq$.
Two ternary predicates $\predfuture$ and $\predright$, that encode the successor relation for time and for cells.
Two binary predicates $\predcopyleft$ and $\predcopyright$, used to copy tapes content.
A unary predicate $\predreal$ and a binary predicate $\pn{NextBr}$, used for the machinery of emergency brakes.
Two unary predicates $\pn{Brake}$ and $\pn{End}$ to identify terms used as emergency brakes and the last element of a configuration, respectively.

Each time a new term is created during the chase, we link it in a specific way to the relevant brake. To simplify the subsequent presentation, we denote by $\brake(x,w)$ the set of atoms $\{\predfuture(x,w,w),\predright(x,w,w),\predreal(x),\predbrake(w)\}$.
The remainder of this section is devoted to the reduction from the ``non-recurrence through \qloop'' problem to knowledge base restricted chase termination. We first present the reduction, and then focus on the main ideas required to show correctness. 

\paragraph{The Reduction}

Each configuration $\rho$ of a Turing machine is encoded by a database as follows.

\begin{definition}
\label{def-reduction-database}
	The database $D_\rho$ encoding a configuration $\rho=\langle n,t,p,q \rangle$ is
	\begin{align*}
		D_{\rho}=&\setst{\predright(c_i,c_{i+1},w_1), \predname{a_i}(c_i,w_1)}{1\leq i\leq n; \predname{a_i} = t(i)} \cup \set{\predq(c_p,w_1), \predname{\blank}(c_{n+1},w_1),\predend(c_{n+1},w_1)}\\
		&\cup \bigcup_{1 \leq i \leq n+1} \brake(c_i,w_1) 
	\end{align*}
  For a word $w$, we denote by $D_w$ the database $D_{\rho_w}$, where $\rho_w$ is the initial configuration of $M$ on $w$.
\end{definition}


Given a Turing machine $M$ with states $Q$ and tape alphabet $\Gamma$, we build $\aruleset_M$ composed of the following rules. We first have a set of rules required for setting up emergency brakes.


\begin{align*}
	\predbrake(w) \to \bigwedge_{\predname{a} \in \Gamma \cup \{\blank\}} \predname{a}(w,w), \bigwedge_{\predname{q} \in Q} \predname{q}(w,w), &\predfuture(w,w,w), \predright(w,w,w), \\
    &\predcopyleft(w,w), \predcopyright(w,w), \predreal(w), \nextBr(w,w) \label{rule:brake} \tag{\ensuremath{R_\predbrake}} \\
  \brake(x,w), \nextBr(w,w') & \to \brake(x,w') \label{rule:nextBr} \tag{\ensuremath{R_\pn{nextBr}}}
\end{align*}

The next four rules are responsible of simulating the moves of the head of the Turing machine. The first two rules deal with the case where the machine is not in \qloop, and the head moves to the right (resp. to the left). The important feature of these rules is the presence in both the body and the head of the same brake $w$.

\begin{align*}
	\intertext{For all $q \neq \qloop, q' \in Q$ and $a,b,c \in \Gamma \cup \{\blank\}$ such that $(q',b,\goright)\in\delta(q,a)$:}
	&\predq(x,w),\preda(x,w),\predright(x,y,w),\predc(y,w),\brake(x,w),\brake(y,w)
	\\
	&\to \exists x',y'~ \predqp(y',w),\predc(y',w),\predb(x',w),\predcopyleft(x',w),\predcopyright(y',w), \label{rule:regular-right}  \tag{\ensuremath{R_{\neg \qloop}^\goright}}\\
	&\qquad\predright(x',y',w),\predfuture(x,x',w),\predfuture(y,y',w),\brake(x',w), \brake(y',w)\\
	\intertext{For all $q \neq \qloop, q' \in Q$ and $a,b,c \in \Gamma \cup \{\blank\}$ such that $(q',b,\goleft)\in\delta(q,a)$:}
	&\predq(x,w),\preda(x,w),\predright(y,x,w),\predc(y,w),\brake(x,w),\brake(y,w)
	\\
	&\to \exists x',y'~ \predqp(y',w),\predc(y',w),\predb(x',w),\predcopyleft(y',w),\predcopyright(x',w), \label{rule:regular-left} \tag{\ensuremath{R_{\neg \qloop}^\goleft}}\\
	&\qquad\predright(y',x',w),\predfuture(x,x',w),\predfuture(y,y',w),\brake(x',w), \brake(y',w)
\end{align*}

The following two rules treat the case where the transition is from \qloop. The only difference with the two above rules is the introduction of a new brake $w'$ in the head of the rules. This permits non-terminating restricted chase sequences in the presence of specific runs. 

\begin{align*}
    \intertext{For all $q' \in Q$ and $a,b,c \in \Gamma \cup \{\blank\}$ such that $(q',b,\goright)\in\delta(\qloop,a)$:}
	&\predqloop(x,w),\predright(x,y,w),\preda(x,w),\predc(y,w),\brake(x,w),\brake(y,w) \\
	&\to \exists x',y',w',\ \predqp(y',w'),\predc(y',w'),\predb(x',w'),\predright(x',y',w'), \label{rule:init-right} \tag{\ensuremath{R_{\qloop}^\goright}}\\
	&\qquad\predfuture(x,x',w'),\predfuture(y,y',w'),\predcopyleft(x',w'),\predcopyright(y',w'),\\
	&\qquad\brake(x',w'),\brake(y',w'), \nextBr(w,w')\\
	\intertext{For all $q' \in Q$ and $a,b,c \in \Gamma \cup \{\blank\}$ such that $(q',b,\goleft)\in\delta(\qloop,a)$:}
	&\predqloop(x,w),\predright(y,x,w),\preda(x,w),\predc(y,w),\brake(x,w),\brake(y,w) \\
	&\to \exists x',y',w',\ \predqp(y',w'),\predc(y',w'),\predb(x',w'),\predright(y',x',w'), \label{rule:init-left}\tag{\ensuremath{R_{\qloop}^\goleft}}\\
	&\qquad\predfuture(x,x',w'),\predfuture(y,y',w'),\predcopyleft(y',w'),\predcopyright(x',w'),\\
	&\qquad\brake(x',w'),\brake(y',w'), \nextBr(w,w')
\end{align*}

The following rules copy the content of unchanged cells to the right and the left of the head from one configuration to the next.
We instantiate one of each rule for each $a \in \Gamma \cup \{\blank\}$.

\begin{align*}
	\predcopyright(x',w'&),\predfuture(x,x',w'),\predright(x,y,w),\preda(y,w),\brake(x,w),\brake(x',w'),\brake(y,w)\\ 
	&\to \exists y'~\ \predfuture(y,y',w'),\predright(x',y',w'),\preda(y',w'),\predcopyright(y',w'),
	\brake(y',w') \label{rule:copyright} \tag{\ensuremath{R_\predcopyright}}\\
  	\predcopyleft(x',w'&),\predfuture(x,x',w'),\predright(y,x,w),\preda(y,w),\brake(x,w),\brake(x',w'),\brake(y,w)\\ 
	&\to \exists y'~\ \predfuture(y,y',w'),\predright(y',x',w'),\preda(y',w'),\predcopyleft(y',w'),
	\brake(y',w') \label{rule:copyleft} \tag{\ensuremath{R_\predcopyleft}}
\end{align*}


Finally, we extend the represented part of the configuration by one cell at each step, as coherent with our definition of Turing machine runs:
\begin{align*}
	&\predcopyright(x',w'),\predfuture(x,x',w'),\predend(x,w),\brake(x,w),\brake(x',w')\\
	&\to \exists y',\ \predright(x',y',w'),\predname{\blank}(y',w'),\predend(y',w'),\brake(y',w') \label{rule:end} \tag{\ensuremath{R_\predend}}
\end{align*}

%

{ 
\definecolor{bleudefrance}{rgb}{0.19, 0.55, 0.91}

\newcommand*{\wsep}{1.8}
\newcommand*{\wa}{-0.7*\wsep}
\newcommand*{\wb}{1*\wsep}
\newcommand*{\wc}{2*\wsep}
\newcommand*{\we}{3*\wsep}
\newcommand*{\wf}{4*\wsep}
\newcommand*{\wg}{5*\wsep}
\newcommand*{\wh}{6*\wsep}
\newcommand*{\wi}{7*\wsep}

\newcommand*{\vsep}{0.5}
\newcommand*{\va}{12*\vsep}
\newcommand*{\vb}{11*\vsep}
\newcommand*{\vc}{10*\vsep}
\newcommand*{\ve}{9*\vsep}
\newcommand*{\vf}{8*\vsep}
\newcommand*{\vg}{7*\vsep}
\newcommand*{\vh}{6*\vsep}
\newcommand*{\vi}{5*\vsep}
\newcommand*{\vj}{4*\vsep}
\newcommand*{\vk}{3*\vsep}
\newcommand*{\vl}{2*\vsep}
\newcommand*{\vm}{1*\vsep}
\newcommand*{\vn}{0*\vsep}

\begin{figure}
\begin{tikzpicture}
[term/.style={draw, fill=black, circle, inner sep=1pt},
termLabel/.style={inner sep=1pt, outer sep=3pt, fill=white},
binary/.style={->, thick},
arrowLabel/.style={inner sep=1pt, fill=white}]

\draw[dashed, bleudefrance, line width=1.5pt] (-0.9*\wsep, \va+\vsep/2) rectangle  (\wc+\wsep/3, \vb+\vsep/2);
\draw[dashed, bleudefrance, line width=1.5pt] (-0.9*\wsep, \vc+\vsep/2) rectangle  (\wf+\wsep/3, \vf-\vsep/2);
\draw[dashed, bleudefrance, line width=1.5pt] (-0.9*\wsep, \vh+\vsep/2) rectangle  (\wh+\wsep/3, \vj-\vsep/2);

\node[term] (ba) at (\wb, \va) {};
\node[termLabel, above right] at (ba) {\pn{\qinit},\pn{0}};
\node[term] (ba) at (\wb, \va) {};
\node[term] (ca) at (\wc, \va) {};
\node[termLabel, above right] at (ca) {\pn{\blank},\predend};
\node[term] (ca) at (\wc, \va) {};

\node[term] (bc) at (\wb, \vc) {};
\node[termLabel, above right] at (bc) {\pn{1},\predcopyleft};
\node[term] (bc) at (\wb, \vc) {};
\node[term] (cc) at (\wc, \vc) {};
\node[termLabel, above right] at (cc) {\pn{\qloop},\pn{\blank},\predcopyright};
\node[term] (cc) at (\wc, \vc) {};
\node[term] (ec) at (\we, \vc) {};
\node[termLabel, above right] at (ec) {\pn{\blank},\predend};
\node[term] (ec) at (\we, \vc) {};

\node[term] (bf) at (\wb, \vf) {};
\node[termLabel, above right] at (bf) {\pn{\qinit},\pn{1},\predcopyleft};
\node[term] (bf) at (\wb, \vf) {};
\node[term] (cf) at (\wc, \vf) {};
\node[termLabel, above right] at (cf) {\pn{1},\predcopyright};
\node[term] (cf) at (\wc, \vf) {};
\node[term] (ef) at (\we, \vf) {};
\node[termLabel, above right] at (ef) {\blank,\predcopyright};
\node[term] (ef) at (\we, \vf) {};
\node[term] (ff) at (\wf, \vf) {};
\node[termLabel, above right] at (ff) {\pn{\blank},\predend};
\node[term] (ff) at (\wf, \vf) {};

\node[term] (bh) at (\wb, \vh) {};
\node[termLabel, above right] at (bh) {\pn{1},\predcopyleft};
\node[term] (bh) at (\wb, \vh) {};
\node[term] (ch) at (\wc, \vh) {};
\node[termLabel, above right] at (ch) {\pn{\qloop},\pn{1},\predcopyright};
\node[term] (ch) at (\wc, \vh) {};
\node[term] (eh) at (\we, \vh) {};
\node[termLabel, above right] at (eh) {\blank,\predcopyright};
\node[term] (eh) at (\we, \vh) {};
\node[term] (fh) at (\wf, \vh) {};
\node[termLabel, above right] at (fh) {\blank,\predcopyright};
\node[term] (fh) at (\wf, \vh) {};
\node[term] (gh) at (\wg, \vh) {};
\node[termLabel, above right] at (gh) {\pn{\blank},\predend};
\node[term] (gh) at (\wg, \vh) {};

\node[term] (bj) at (\wb, \vj) {};
\node[termLabel, above right] at (bj) {\pn{\qinit},\pn{1},\predcopyleft};
\node[term] (bj) at (\wb, \vj) {};
\node[term] (cj) at (\wc, \vj) {};
\node[termLabel, above right] at (cj) {\pn{1},\predcopyright};
\node[term] (cj) at (\wc, \vj) {};
\node[term] (ej) at (\we, \vj) {};
\node[termLabel, above right] at (ej) {\pn{\blank},\predcopyright};
\node[term] (ej) at (\we, \vj) {};
\node[term] (fj) at (\wf, \vj) {};
\node[termLabel, above right] at (fj) {\pn{\blank},\predcopyright};
\node[term] (fj) at (\wf, \vj) {};
\node[term] (gj) at (\wg, \vj) {};
\node[termLabel, above right] at (gj) {\pn{\blank},\predcopyright};
\node[term] (gj) at (\wg, \vj) {};
\node[term] (hj) at (\wh, \vj) {};
\node[termLabel, above] at (hj) {\pn{\blank},\predend};
\node[term] (hj) at (\wh, \vj) {};

\node[term] (ab) at (\wa, \va) {};
\node[termLabel, below right] at (ab) {\pn{Brake}, \pn{Real}};
\node[term] (ab) at (\wa, \va) {};

\node[term] (ag) at (\wa, \ve) {};
\node[termLabel, below right] at (ag) {\pn{Brake}, \pn{Real}};
\node[term] (ag) at (\wa, \ve) {};

\node[term] (ak) at (\wa, \vi) {};
\node[termLabel, below right] at (ak) {\pn{Brake}, \pn{Real}};
\node[term] (ak) at (\wa, \vi) {};

\node (bl) at (\wb, \vl) {};
\node (cl) at (\wc, \vl) {};
\node (el) at (\we, \vl) {};
\node (fl) at (\wf, \vl) {};
\node (gl) at (\wg, \vl) {};
\node (hl) at (\wh, \vl) {};



\path[binary] (ba) edge node[arrowLabel, near end] {\pn{R}} (ca);

\path[binary] (bc) edge node[arrowLabel, near end] {\pn{R}} (cc);
\path[binary] (cc) edge node[arrowLabel, near end] {\pn{R}} (ec);

\path[binary] (bf) edge node[arrowLabel, near end] {\pn{R}} (cf);
\path[binary] (cf) edge node[arrowLabel, near end] {\pn{R}} (ef);
\path[binary] (ef) edge node[arrowLabel, near end] {\pn{R}} (ff);

\path[binary] (bh) edge node[arrowLabel, near end] {\pn{R}} (ch);
\path[binary] (ch) edge node[arrowLabel, near end] {\pn{R}} (eh);
\path[binary] (eh) edge node[arrowLabel, near end] {\pn{R}} (fh);
\path[binary] (fh) edge node[arrowLabel, near end] {\pn{R}} (gh);

\path[binary] (bj) edge node[arrowLabel, near end] {\pn{R}} (cj);
\path[binary] (cj) edge node[arrowLabel, near end] {\pn{R}} (ej);
\path[binary] (ej) edge node[arrowLabel, near end] {\pn{R}} (fj);
\path[binary] (fj) edge node[arrowLabel, near end] {\pn{R}} (gj);
\path[binary] (gj) edge node[arrowLabel, near end] {\pn{R}} (hj);


\path[binary] (ba) edge node[arrowLabel,yshift=1.5mm] {\pn{F}} (bc);
\path[binary] (ca) edge node[arrowLabel,yshift=1.5mm] {\pn{F}} (cc);

\path[binary] (bc) edge node[arrowLabel,yshift=1.5mm] {\pn{F}} (bf);
\path[binary] (cc) edge node[arrowLabel,yshift=1.5mm] {\pn{F}} (cf);
\path[binary] (ec) edge node[arrowLabel,yshift=1.5mm] {\pn{F}} (ef);

\path[binary] (bf) edge node[arrowLabel,yshift=1.5mm] {\pn{F}} (bh);
\path[binary] (cf) edge node[arrowLabel,yshift=1.5mm] {\pn{F}} (ch);
\path[binary] (ef) edge node[arrowLabel,yshift=1.5mm] {\pn{F}} (eh);
\path[binary] (ff) edge node[arrowLabel,yshift=1.5mm] {\pn{F}} (fh);

\path[binary] (bh) edge node[arrowLabel,yshift=1.5mm] {\pn{F}} (bj);
\path[binary] (ch) edge node[arrowLabel,yshift=1.5mm] {\pn{F}} (cj);
\path[binary] (eh) edge node[arrowLabel,yshift=1.5mm] {\pn{F}} (ej);
\path[binary] (fh) edge node[arrowLabel,yshift=1.5mm] {\pn{F}} (fj);
\path[binary] (gh) edge node[arrowLabel,yshift=1.5mm] {\pn{F}} (gj);

\path[binary,dotted] (bj) edge node[arrowLabel,yshift=1.5mm] {\pn{F}} (bl);
\path[binary,dotted] (cj) edge node[arrowLabel,yshift=1.5mm] {\pn{F}} (cl);
\path[binary,dotted] (ej) edge node[arrowLabel,yshift=1.5mm] {\pn{F}} (el);
\path[binary,dotted] (fj) edge node[arrowLabel,yshift=1.5mm] {\pn{F}} (fl);
\path[binary,dotted] (gj) edge node[arrowLabel,yshift=1.5mm] {\pn{F}} (gl);
\path[binary,dotted] (hj) edge node[arrowLabel,yshift=1.5mm] {\pn{F}} (hl);


\path[binary] (ab) edge [out=90,in=0,loop, distance=40] node [fill=white, pos=0.6, inner sep=1pt] {$\pn{AllPreds}_{\geq2}$} (ab);
\path[binary] (ag) edge [out=90,in=0,loop, distance=40] node [fill=white, pos=0.6, inner sep=1pt] {$\pn{AllPreds}_{\geq2}$} (ag);
\path[binary] (ak) edge [out=90,in=0,loop, distance=40] node [fill=white, pos=0.6, inner sep=1pt] {$\pn{AllPreds}_{\geq2}$} (ak);

\path[binary, out=225, in=135] (ab) edge node[arrowLabel, yshift=1.5mm] {\nextBr} (ag);
\path[binary, out=225, in=135] (ag) edge node[arrowLabel] {\nextBr} (ak);

\end{tikzpicture}
\caption{An Infinite Restricted Chase Sequence of $\langle \aruleset_M, D_{0}\rangle$ where $M$ is the machine from \cref{example:kb-reduction}, and $\pn{AllPreds}_{\geq2}$ above is a shortcuts for ``\pn{F}, \pn{R}, \pn{\qinit}, \pn{\qloop}, \pn{0}, \pn{1}, \pn{\blank}, \predcopyright, \predcopyleft, \nextBr''.}
\label{figure:restricted-sequence}
\end{figure}
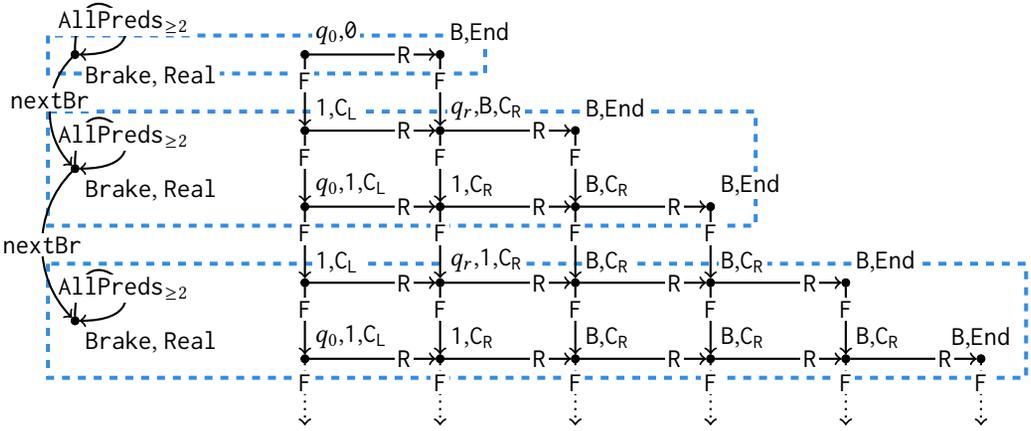

\begin{example}
\label{example:kb-reduction}
Consider a machine $M = \langle \{\qinit, \qloop\}, \{0,1\}, \delta \rangle$ where $\delta$ is a transition function that maps $\langle \qinit, 0\rangle$ to $\{\langle \qloop, 1, \goright \rangle\}$, $\langle \qloop, \blank \rangle$ to $\{\langle \qinit, 1, \goleft \rangle\}$, $\langle \qinit, 1 \rangle$ to $\{\langle \qloop, 1, \goright \rangle\}$, and $\langle \qloop, 1 \rangle$ to $\{\langle \qinit, 1, \goleft \rangle\}$; note how the (only) run of $M$ on the word $0$ contains infinitely many configurations with the state \qloop.
In this representation, every label on an edge or a term represents several facts in the chase.
For the sake of clarity, these labels can be extended with another argument, which should be some ``\pn{Brake}'' term in the same dashed or later dashed box.
\end{example}

} 

\paragraph{Correctness proof of the reduction}

The reduction is now fully described, and we claim that:

\begin{proposition}
\label{prop:kb-reduction-correct}
 $\aruleset_M$ universally halts for the restricted chase on $D_\rho$ if and only if there exists no run of $M$ on $\rho$ that goes infinitely often through \qloop. 
\end{proposition}

We first prove that if there exists a run of $M$ going through \qloop infinitely often, then there exists a non-terminating chase sequence. To that purpose, we identify interesting subsets of databases.

\begin{definition}[Wild Frontier of Configuration $\rho$]
\label{def-wild-frontier}
 A set of atoms $F$ has a \emph{wild frontier} of configuration 
 $\rho=\langle n,t,p,q\rangle$ overseen by $w \in \terms(F)$ if there exists $x_1,\ldots,x_{n+1} \in \terms(F)$ such that:
 \begin{itemize}
  \item $\predreal(w) \not \in F$;
  \item $\{\predright(x_i,x_{i+1},w),\predname{a_i}(x_i,w)\} \subseteq F$ for all $i \in \{1,\ldots,n\}$, $\predname{a_i} = t(i)$;
  \item $\predq(x_p,w), \predend(x_{n+1},w), \predname{\blank}(x_{n+1},w) \in F$;
  \item $\brake(x_i,w) \in F$ for all $i \in \{1,\ldots,n+1\}$;
  \item any other atom of $F$ having $x_i$ as first argument has $w$ as second.
 \end{itemize}

\end{definition}

A wild frontier has three important features \emph{(i)} it contains the necessary atoms to simulate the run of a Turing machine on that configuration; \emph{(ii)} it is correctly connected to a (not yet real) brake $w$; \emph{(iii)} it does not contain atoms preventing the above run to be simulated through a restricted derivation. By comparing Definition~\ref{def-reduction-database} and Definition~\ref{def-wild-frontier}, it is clear that $D_\varepsilon$ has a wild frontier of the configuration of $M$ on the empty word, overseen by $w_1$.
The construction of an infinite restricted derivation is made by inductively using the following key proposition. 

\begin{propositionrep}
\label{prop:frontier-respawn}
 If $F$ has a wild frontier of $\rho$ overseen by $w$, and $\rho'$  is reachable in one step by a transition of $M$, then there exists a restricted derivation $\der_{\rho\to\rho'} = F,\ldots, F'$ such that $F'$ has a wild frontier of $\rho'$ overseen by $w'$, where $w' \not =  w$ is a fresh existential if $\rho$ is in \qloop, and $w' = w$ otherwise. 
\end{propositionrep}

\begin{proof} 
We consider the case where $\rho = \langle n,t,p,q\rangle$, with $q \not = \qloop$, and where $\rho'$ is obtained from $\rho$ because $(b,q',\goright) \in \delta(t(p),q)$. We consider $x_1,\ldots,x_{n+1},w$ as provided by the definition of a wild frontier of configuration $\rho$. 
\begin{itemize}
\item we start by applying Rule~\ref{rule:regular-right}, mapping $x$ to $x_p$, $y$ to $x_{p+1}$ and $w$ to $w$. This produces the atoms $\predq'(x'_{p+1},w)$, $\predname{c}(x'_{p+1},w)$, $\predname{b}(x'_p,w)$, $\predcopyleft(x'_p,w)$, $\predcopyright(x'_{p+1},w)$, $\predright(x'_p,x'_{p+1},w)$, $\predfuture(x_p,x'_p,w)$, $\predfuture(x_{p+1},x'_{p+1},w)$, $\brake(x'_p,w)$, $\brake(x'_{p+1},w)$;
\item we apply Rule~\ref{rule:copyleft} $p-1$ times. The $i^\mathrm{th}$ (for $i$ from $1$ to $p-1$) application maps $x$ to $x_{p-i+1}$, $x'$ to $x'_{p-i+1}$, $y$ to $x_{p-i}$, $w'$ and $w$ to $w$. It creates atoms $\predfuture(x_{p-i},x'_{p-i},w)$, $\predright(x'_{p-i},x'_{p-i+1},w)$, $\predname{t(p-i)}(x'_{p-i},w)$, $\predcopyleft(x'_{p-i},w)$, $\brake(x'_{p-i},w')$;
\item we apply Rule~\ref{rule:copyright} $n-p$ times. The $i^\mathrm{th}$ (for $i$ from $1$ to $n-p$) application maps $x$ to $x_{p+i}$, $x'$ to $x'_{p+i}$, $y$ to $x_{p+i+1}$, $w'$ and $w$ to $w$. It creates atoms $\predfuture(x_{p+i+1},x'_{p+i+1},w)$, $\predright(x'_{p+i},x'_{p+i+1},w)$, $\predname{t(p+i+1)}(x'_{p+i+1},w)$, $\predcopyright(x'_{p+i+1},w)$, $\brake(x'_{p+i+1},w)$
\item we apply Rule~\ref{rule:end}, mapping $x'$ to $x'_{n+1}$, $x$ to $x_{n+1}$, $w$ and $w'$ to $w$. It creates the atoms $\predright(x'_{n+1},x'_{n+2},w)$, $\predname{\blank}(x'_{n+2},w)$, $\predend(x'_{n+2},w)$,$\brake(x'_{n+2},w)$.
\end{itemize}
The result of that derivation has a wild frontier of configuration $\rho'$ overseen by $w$, as witnessed by terms $x'_1,\ldots,x'_{n+2}$.\\

If $\rho'$ is obtained from $\rho$ because $(b,q',\goleft) \in \delta(t(p),q)$, with $q \not = \qloop$, we consider $x_1,\ldots,x_{n+1},w$ as provided by the definition of a wild frontier of configuration $\rho$.
\begin{itemize}
 \item we start by applying Rule~\ref{rule:regular-left}, mapping $x$ to $x_p$, $y$ to $x_{p-1}$, $w$ to $w$. This produces the atoms $q'(x'_{p-1},w)$, $\predc(x'_{p-1},w)$, $\predb(x'_p,w)$, $\predcopyleft(x'_{p-1},w)$, $\predcopyright(x'_p,w)$, $\predright(x'_{p-1},x'_p,w)$, $\predfuture(x_p,x'_p,w)$, $\predfuture(x_{p-1},x'_{p-1},w)$, $\brake(x'_p,w)$, $\brake(x'_{p-1},w)$;
 \item we apply Rule~\ref{rule:copyleft} $p-2$ times. The $i^\mathrm{th}$ (for $i$ from $1$ to $p-2$) application maps $x$ to $x_{p-i}$, $x'$ to $x'_{p-1}$, $y$ to $x_{p-i-1}$, and $w,w'$ to $w$. It creates atoms $\predfuture(x_{p-i-1},x'_{p-i-1},w)$, $\predright(x'_{p-i-1}, x'_{p-i},w)$, $\predname{t(p-i-1)}(x'_{p-i-1},w)$, $\predcopyleft(x_{p-i-1},w)$, $\brake(x_{p-i-1},w)$;
 \item we apply Rule~\ref{rule:copyright} $n-p+1$ times. The $i^\mathrm{th}$ (for $i$ from $1$ to $n-p+1$) application maps $x$ to $x_{p-1+i}$, $x'$ to $x'_{p-1+i}$, $y$ to $x_{p+i}$, and $w,w'$ to $w$. It creates atoms $\predfuture(x_{p+i},x'_{p+i},w)$, $\predright(x'_{p+i-1},x'_{p+i},w)$, $\predname{t(p-i-1)}(x'_{p+i},w)$, $\predcopyright(x'_{p+i},w)$, $\brake(x'_{p+i},w)$;
 \item we apply Rule~\ref{rule:end}, mapping $x'$ to $x'_{n+1}$, $x$ to $x_{n+1}$, $w$ and $w'$ to $w$. It creates the atoms $\predright(x'_{n+1},x'_{n+2},w)$, $\predname{\blank}(x'_{n+2},w)$, $\predend(x'_{n+2},w)$,$\brake(x'_{n+2},w)$.
\end{itemize}
The result of that derivation has a wild frontier of configuration $\rho'$ overseen by $w$, as witnessed by terms $x'_1,\ldots,x'_{n+2}$.\\

If $\rho'$ is obtained from $\rho$ because $(b,q',\goright) \in \delta(t(p),\qloop)$, we consider $x_1,\ldots,w_{n+1}$ as provided by the definition of a wild frontier of configuration $\rho$.

\begin{itemize}
\item we start by applying Rule~\ref{rule:init-right}, mapping $x$ to $x_p$, $y$ to $x_{p+1}$ and $w$ to $w$. This rule application produces the atoms $\predq'(x'_{p+1},w')$, $\predname{c}(x'_{p+1},w')$, $\predname{b}(x'_p,w')$, $\predcopyleft(x'_p,w')$, $\predcopyright(x'_{p+1},w')$, $\predright(x'_p,x'_{p+1},w')$, $\predfuture(x_p,x'_p,w')$,$\predfuture(x_{p+1},x'_{p+1},w')$, $\brake(x'_p,w')$, $\brake(x'_{p+1},w')$;
\item we apply Rule~\ref{rule:copyleft} $p-1$ times. The $i^\mathrm{th}$ (for $i$ from $1$ to $p-1$) application maps $x$ to $x_{p-i+1}$, $x'$ to $x'_{p-i+1}$, $y$ to $x_{p-i}$, $w$ to $w$ and $w'$ to $w'$. It creates atoms $\predfuture(x_{p-i},x'_{p-i},w')$, $\predright(x'_{p-i},x'_{p-i+1},w')$, $\predname{t(p-i)}(x'_{p-i},w')$, $\predcopyleft(x'_{p-i},w')$, $\brake(x'_{p-i},w')$;
\item we apply Rule~\ref{rule:copyright} $n-p$ times. The $i^\mathrm{th}$ (for $i$ from $1$ to $n-p$) application maps $x$ to $x_{p+i}$, $x'$ to $x'_{p+i}$, $y$ to $x_{p+i+1}$, $w$ to $w$ and $w'$ to $w'$. It creates atoms $\predfuture(x_{p+i+1},x'_{p+i+1},w')$, $\predright(x'_{p+i},x'_{p+i+1},w')$, $\predname{t(p+i+1)}(x'_{p+i+1},w')$, $\predcopyright(x'_{p+i+1},w')$, $\brake(x'_{p+i+1},w')$
\item we apply Rule~\ref{rule:end}, mapping $x'$ to $x'_{n+1}$, $x$ to $x_{n+1}$, $w$ to $w$ and $w'$ to $w'$. It creates atoms $\predright(x'_{n+1},x'_{n+2},w')$, $\predname{\blank}(x'_{n+2},w')$, $\predend(x'_{n+2},w')$,$\brake(x'_{n+2},w')$.
\end{itemize}

The result of that derivation has a wild frontier of configuration $\rho'$ overseen by $w'$, as witnessed by terms $x'_1,\ldots,x'_{n+2}$.\\

If $\rho'$ is obtained from $\rho$ because $(b,q',\goleft) \in \delta(t(p),\qloop)$, we consider $x_1,\ldots,x_{n+1},w$ as provided by the definition of a wild frontier of configuration $\rho$.
\begin{itemize}
 \item we start by applying Rule~\ref{rule:init-left}, mapping $x$ to $x_p$, $y$ to $x_{p-1}$, $w$ to $w$. This rule application produces the atoms $q'(x'_{p-1},w')$, $\predc(x'_{p-1},w')$, $\predb(x'_p,w')$, $\predcopyleft(x'_{p-1},w')$, $\predcopyright(x'_p,w')$, $\predright(x'_{p-1},x'_p,w')$, $\predfuture(x_p,x'_p,w'), \predfuture(x_{p-1},x'_{p-1},w')$, $\brake(x'_p,w')$, $\brake(x'_{p-1},w')$;
 \item we apply Rule~\ref{rule:copyleft} $p-2$ times. The $i^\mathrm{th}$ (for $i$ from $1$ to $p-2$) application maps $x$ to $x_{p-i}$, $x'$ to $x'_{p-1}$, $y$ to $x_{p-i-1}$, and $w$ to $w$ and $w'$ to $w'$. It creates atoms $\predfuture(x_{p-i-1},x'_{p-i-1},w')$, $\predright(x'_{p-i-1}, x'_{p-i},w')$, $\predname{t(p-i-1)}(x'_{p-i-1},w')$, $\predcopyleft(x_{p-i-1},w')$, $\brake(x_{p-i-1},w')$;
 \item we apply Rule~\ref{rule:copyright} $n-p+1$ times. The $i^\mathrm{th}$ (for $i$ from $1$ to $n-p+1$) application maps $x$ to $x_{p-1+i}$, $x'$ to $x'_{p-1+i}$, $y$ to $x_{p+i}$, and $w$ to $w$ and $w'$ to $w'$. It creates atoms $\predfuture(x_{p+i},x'_{p+i},w')$, $\predright(x'_{p+i-1},x'_{p+i},w')$, $\predname{t(p-i-1)}(x'_{p+i},w')$, $\predcopyright(x'_{p+i},w')$, $\brake(x'_{p+i},w')$;
 \item we apply Rule~\ref{rule:end}, mapping $x'$ to $x'_{n+1}$, $x$ to $x_{n+1}$, $w$ to $w$ and $w'$ to $w'$. It creates atoms $\predright(x'_{n+1},x'_{n+2},w')$, $\predname{\blank}(x'_{n+2},w')$, $\predend(x'_{n+2},w')$,$\brake(x'_{n+2},w')$.
\end{itemize}
The result of that derivation has a wild frontier of configuration $\rho'$ overseen by $w'$, as witnessed by terms $x'_1,\ldots,x'_{n+2}$.

\end{proof}

Concatenating the infinite sequence of derivations built in Proposition~\ref{prop:frontier-respawn} does not however provide a fair sequence of derivations, because of Rules~\ref{rule:brake}, \ref{rule:nextBr} and of the non-determinism of $M$. Fairness is enforced by applying \ref{rule:brake} and \ref{rule:nextBr} ``late enough'' to ensure that none of the triggers involved in the proof of Proposition~\ref{prop:frontier-respawn} are made obsolete. This is possible because the run of $M$ going infinitely many often through \qloop, infinitely many brakes are created. Details are provided in the appendix.

\begin{toappendix}

A rule is \emph{datalog} if its head does not contain any existentially quantified variable. 

\begin{proposition}
\label{prop:brake}
Let $F_0,\ldots,F_k$ be restricted derivation.  Let $w^* \in \terms(F_k) \setminus \terms(F_0)$ such that $\predreal(w^*) \in F_k$. Then for any $j > k$, the only rules generating $w^*$ as a last argument having non-obsolete triggers on $F_j$ are datalog rules. 
\end{proposition}

\begin{proof}
We consider a non-datalog rule $\arule$ and $\sigma$ a homomorphism of $\body(\arule)$ into $F_j$. We prove that $\sigma$ can be extended in a homomorphism of $\head(\arule)$ into $F_j$, showing that $\langle \arule,\sigma\rangle$ is obsolete. Note that as $\predreal(w^*) \in F_k$, it holds that Rule~\ref{rule:brake} has been applied by mapping $w$ to $w^*$.  

\begin{itemize}
 \item Rules~\ref{rule:regular-left} and \ref{rule:regular-right}: extend $\sigma$ by mapping $x'$ and $y'$ to $\sigma(w) = w^*$.
 \item Rules~\ref{rule:init-left} and \ref{rule:init-right}: extend $\sigma$ by mapping $x',y', w'$ to $\sigma(w) = w^*$
 \item Rules~\ref{rule:copyleft}, \ref{rule:copyright}, and \ref{rule:end}: extend $\sigma$ by mapping $y'$ to $\sigma(w') = w^*$.
\end{itemize}

\end{proof}
\end{toappendix}

\begin{lemmarep}
\label{lemma:TM_to_chase}
Let $(\rho_i)_{i \in \mathbb{N}}$ be a run of $M$ on the empty word that visits \qloop infinitely often. There exists an infinite restricted chase sequence for $\langle\aruleset_M,D_\varepsilon\rangle$. 
\end{lemmarep}

\begin{proof}
 Let $(i_j)_{j\in \mathbb{N}}$ be the infinite strictly increasing sequence of integers such that $i_1 = 1$ and $\rho_k$ is in \qloop if and only if $k = i_j$ for some $j$. We denote by $\der_{\rho_{i_j} \rightarrow \rho_{i_{j+1}}}$ the concatenation of the restricted derivations provided by Proposition~\ref{prop:frontier-respawn}. Let us consider the derivation build by induction:
 \begin{itemize}
   \item $\der_1 = \der_{\rho_{i_1} \rightarrow \rho_{i_2}}$
   \item $\der'_1$ extends $\der_1$ by the application of Rule~\ref{rule:brake} mapping $w$ to the brake overseeing the wild frontier of the last element of $\der_1$, as well as by applying any datalog rule mapping $w$ to that brake.
   \item $\der_{j}$ extends $\der'_{j-1}$ by the derivation $\der_{\rho_{i_j} \rightarrow \rho_{i_{j+1}}}$;
   \item $\der'_j$ extends $\der_j$ by the application of Rule~\ref{rule:brake} mapping $w$ to the brake overseeing the wild frontier of the last element of $\der_j$, and by applying Rule~\ref{rule:nextBr} in any possible way that maps $w$ to the brake overseeing the wild frontier of the last element of $\der_j$, as well by applying any datalog rule mapping $w$ to that brake.
  \end{itemize}

This derivation is fair:
\begin{itemize}
    \item any created atom has a brake as argument;
    \item brakes are created exactly once in each derivation $\der_{\rho_{i_j} \rightarrow \rho_{i_{j+1}}}$ (by definition of $(i_j)_{j \in \mathbb{N}}$); let us call $w_1$ the brake appearing in $D_\varepsilon$, and  $w_{j+1}$ the brake crated in $\der_{\rho_{i_j} \rightarrow \rho_{i_{j+1}}}$;
    \item by Proposition~\ref{prop:brake}, the application of Rule~\ref{rule:brake} mapping $w$ to $w_{j}$ deactivates any trigger of a non-datalog rule mapping creating an atom with $w_j$ as a last argument;
    \item by definition of $\der'_j$, all datalog rules creating an atom with $w_j$ as last argument are applied.
\end{itemize}


\end{proof}



To show the converse, we fix an infinite restricted chase sequence $\der$ as $(F_i)_{i \in \mathbb{N}}$, where $F_0 = D_\varepsilon$. We build from $\der$ an infinite run that visits \qloop infinitely often by identifying a substructure of the chase, consisting of \emph{state atoms}. We then prove that a run can be built from these states atoms (and other elements of the chase), which fulfills the required conditions.

\begin{definition}
	A \emph{state atom of $F$} is an atom of $F$ of the form $\predq(x,w)$ where $q\in Q$ and $x$ is not a brake. 
    A state atom $A$ \emph{precedes} $A'$ if there is a trigger $t$ such that $A\in\trigsupport(t)$ and $A'\in\trigoutput(t)$. In this case, we write $A\prec A'$.
\end{definition}

It is worth noticing that in the chase of $\langle\aruleset_M,D_\varepsilon\rangle$, state atoms are organised as a tree structure rooted in the unique state atom belonging to $D_\varepsilon$, and such that $A$ is a parent of $A'$ if and only if $A \prec A'$. Intuitively, we can assign with each of these state atoms a configuration such that the configuration associated with $A'$ is reachable in one transition of $M$ from the configuration associated with its parent $A$. The key property is that in an infinite restricted chase, there exists an infinite sequence $(A_n)_{n\in\nat}$ with good properties.

\begin{toappendix}
  \begin{lemma}\label{lem:finitely_many_atoms_per_brake}
    Let $F$ be the result of an infinite restricted chase sequence $F_0,F_1,\dots$ from $\angles{\aruleset_M,D}$ for some $D$. For any $w$ such that $\predbrake(w) \in F$, there are finitely many atoms having $w$ as last argument in $F$. There is thus an infinite amount of brakes in $F$.
  \end{lemma}
  \begin{proof}
    Consider a term $w$ such that $\predbrake(w)\in F$, which we call a brake. By fairness and Rule~\ref{rule:brake}, there must be some integer $i$ such that $\predreal(w)\in F_i$. At this step, there is a finite number of atoms with $w$ as last argument, and by \cref{prop:brake}, the only rules that can generate such atoms after step $i$ are datalog. Rule~\ref{rule:brake} only generates atoms over $w$, so it is applicable at most one, and will yield at most 6 new atoms. Thus, the only rule left is Rule~\ref{rule:nextBr}.

    Only two rules create new $\predbrake$-atoms that do not already appear in their bodies, which are Rules \ref{rule:init-right} and \ref{rule:init-left}. Both these rules also generate an atom of the form $\nextBr(w,w')$, where $\predbrake(w)$ is the brake in their body, and $\predbrake(w')$ is the newly created brake. As this is the only way to generate $\nextBr$-atoms, the predicate $\nextBr$ defines a forest relationship over the brakes, where the root of each tree is a term $w_0$ such that $\predbrake(w_0)\in D$. There is thus a finite number of trees. We then show that Rule~\ref{rule:nextBr} can only create a finite number of atoms by induction on this forest structure.
    \begin{itemize}
      \item If $\predbrake(w)\in D$, then all the atoms of the form $\nextBr(w',w)$ are in $D$, so $w'$ is in $D$ too. Thus, Rule~\ref{rule:nextBr} can only create sets of atoms of the form $\brake(x,w')$, where $x$ is a database term. As there is a finite amount of database terms, this yields a finite number of atoms.
      \item If $\predbrake(w')\in F\setminus D$, $\nextBr(w,w')\in F$ and there is a finite number of atoms having $w$ as last argument, then first notice that $w$ is the only term such that $\nextBr(w,w')\in F$, since Rules \ref{rule:init-right} and \ref{rule:init-left} both generate $\nextBr$-atoms featuring an existential variable in second position. Then, as there is a finite amount of atoms featuring $w$ as their last argument, there is a finite amount of terms $x$ such that $\brake(x,w)\subseteq F$. Thus, Rule~\ref{rule:nextBr} generates at most $\brake(x,w')$ for all these terms, which represents a finite number of atoms.
    \end{itemize}
    Thus, there is a finite number of atoms that feature a given brake as their last argument. As $F_0,F_1,\dots$ is infinite, $F$ must have an infinite amount of atoms, that were generated during the chase. Since $\predbrake(w)$ is required in the body of all the rules where $w$ appears as the last argument of an atom, there is thus an infinite amount of brakes in $F$.
  \end{proof}
\end{toappendix}

\begin{lemmarep}\label{lem:infinite_sequence_of_state_atoms}
	For all databases $D$, and all infinite chase sequences from $\angles{\aruleset_M,D}$ with result $F$, there is an infinite sequence $(A_n)_{n\in\nat}$ of state atoms of $F$ such that:
	\begin{itemize}
	 \item $A_0 \in D$;
	 \item $A_n\prec A_{n+1}$ for all $n\in\nat$;
	 \item for infinitely many $i \in \nat$, $A_i$ is of the shape $\predqloop(t_i,w_i)$.
	\end{itemize}
\end{lemmarep}
\begin{proofsketch}
  Since the rules that introduce state atoms (Rules \ref{rule:init-left}, \ref{rule:init-right}, \ref{rule:regular-left} and \ref{rule:regular-right}) feature a state atom in their body, $\prec$ defines a forest structure over state atoms, where the root of each tree is an atom of the database. There is thus a finite amount of trees. We can prove by induction that there is a finite amount of atoms that feature a given brake. Thus, there is an infinite amount of brakes in $F$. Then, since the rules that introduce new brakes (Rules \ref{rule:init-right} and \ref{rule:init-left}) introduce a state atom too, there is an infinite number of state atoms. Thus, one of the trees must be infinite, and since branching can be proven to be finite, there must be an infinite branch by K\"onig's lemma.
\end{proofsketch}
\begin{proof}
  Since the rules that introduce state atoms (Rules \ref{rule:init-left}, \ref{rule:init-right}, \ref{rule:regular-left} and \ref{rule:regular-right}) feature a state atom in their body, $\prec$ defines a forest structure over state atoms, where the root of each tree is an atom of the database. There is thus a finite amount of trees (as there is a finite amount of atoms in the database). By \cref{lem:finitely_many_atoms_per_brake}, there is an infinite amount of brakes in $F$. Then, since the rules that introduce new brakes (Rules \ref{rule:init-right} and \ref{rule:init-left}) introduce a state atom too, there is an infinite number of state atoms. Thus, one of the trees must be infinite.
  In addition, since there is a finite amount of atoms that feature a given brake as last argument, and each state atom features a brake as last argument, each state atom only has a finite number of successors for $\prec$. Indeed, infinitely many successors would require infinitely many rule applications, and thus infinitely many atoms featuring the same last argument as the state atom. We thus have an infinite tree with finite branching. It thus features an infinite branch, which must contain infinitely many $\predqloop$-atoms (as there are infinitely many $\predqloop$-atoms), by K\"onig's lemma.
\end{proof}

%

\begin{toappendix}
To each state atom, we associate both a set of atoms and a configuration. 

\begin{definition}[Atoms associated with a state atom]
 Let $F_k$ be a fact set occuring in a chase derivation from $D_\rho$. The atoms associated with a state atom $\predq(x,w)$ in $F_k$ is the largest subset of $F_k$ whose terms are included in $\{x,w\} \cup X_i$ and such that:
 \begin{itemize}
  \item $X_i$ is the set of terms reachable or co-reachable through an $R$-path from $x$ not going through a brake;
  \item $w$ can only appear in the last position of atoms.
 \end{itemize}
\end{definition}

\begin{definition}[Configuration of a state atom] 
 Let $F_k$ appearing in a restricted chase sequence for $\langle\aruleset,D_\rho\rangle$. The configuration associated with a state atom $\predq(x,w)$, written $\conf(\predq(x,w))$, is defined by induction:
 \begin{itemize}
  \item if $\predq(x,w) \in D_\rho$, $\conf(\predq(x,w)) = \rho$
  \item otherwise, let $A$ be the unique state atom such that $A \prec \predq(x,w)$. Let $(q'b,d)$ be the element of $\delta(q,a)$ due to which the rule whose application generated $\predq(w,x)$ belongs to $\aruleset_M$. We define $\conf(\predq(x,w))$ as the configuration obtained from $\conf(A)$ where the content of the head cell of $\conf(A)$ is replaced by $b$, the head moves by $d$ and switches to state $q'$.
 \end{itemize}
\end{definition}

Note that the above definition implies that if $A_p$ is the parent of $A$, then $\conf(A)$ is reachable in one transition of $M$ from $\conf(A_p)$. Intuitively, the configuration of a state atom is encoded by the atoms associated with it. However, the restricted derivation may not have derived all such atoms, hence we consider a weaker notion, that we coin consistency.

\begin{definition}[Consistency]
\label{def:consistency}
 A set of atoms $A$ associated with a state atom is consistent with the configuration $\langle n,t,p,q\rangle$ if:
 \begin{itemize}
  \item there exists $x$ and $w$ such that $\predq(x,w)$ is the only state atom in $A$, and $\conf(x)$ is in state $q$;
  \item $\preda(x,w)$ is the only letter predicate having $x$ as argument in $A$, and $t(p) = a$;
  \item if there is an $R$ path of length $i$ from $x$ to $x'$, and there is an atom $\preda(x',x'')$, then $x'' = w$, $p+i \leq n+1$ and  $t(p+i)=a$;
  \item there exists at most one atom $\predend(x',x'')$ in $A$, and if it exists then $x'' = w$ and there is an $R$-path from $x$ to $x'$ of length $i$, such that $p+i = n+1$;
  \item if there is an $R$ path of length $i$ from $x'$ to $x$, and there is an atom $\preda(x',x'')$, then $x'' = w$, $p-i \geq 1$ and $t(p-i)=a$.
 \end{itemize}
\end{definition}

As expected, the set of atoms associated wtih a state atom is consistent with its configuration, and this allows us to prove Lemma~\ref{lem:RM_simulates}.
\begin{proposition}
\label{prop:consistency}
Let $F_k$ appearing in a restricted chase sequence for $\langle\aruleset,D_\rho\rangle$. For any state atom $A$ of $F_k$, the set of atoms associated with $A$ is consistent with $\conf(A)$.
\end{proposition}
\begin{proof}
 We prove the result by induction.
 If $A$ is a state atom that does not have any parent, then $A \in F_0 = D_\rho$. The set of atoms associated with $A$ is $D_\rho$, which is consistent with the initial configuration of $M$ on $\rho$ by definition, which is configuration of $A$;
 
 Otherwise, let $A = q'(y',w)$ be a state atom of $F_k$. We prove the result assuming that $A$ has been created by the application of Rule~\ref{rule:regular-left}, mapping $x$ to $x_p$, $w$ to $w$ and $y$ to $y_p$. $A_p$, the parent of $A$, is thus of the shape $q(x_p,w)$ (other possible case would be $A$ being created by Rules~\ref{rule:regular-right}, \ref{rule:init-left} or \ref{rule:init-right}, which are treated similarly). It is easy to check that any term reachable from $y'$ by an $R$-path not going through a brake is either created by the same rule application as $y'$, or has been created by an application of Rule~\ref{rule:copyright} mapping $x'$ to a term reachble by an $R$-path from $y'$ (and similarly for terms co-reachable and Rule~\ref{rule:copyleft}). Then:
 \begin{itemize}
  \item if there exists $y'_{-1}$ such that $\predright(y'_{-1},y',w) \in F_k$, then $y'_{-1}$ has been created by the application of Rule~\ref{rule:copyleft}, in which case $\preda(y'_{-1},w)$ is generated if the cell two positions on the left of the head of $\conf(A_p)$ contains an $a$, that is, if the cell one position on the left of the head of $\conf(A)$ contains an $a$; predecessors of $y'$ further away from $y'$ are treated by induction in a similar way.
  \item there exists a $y'_{+1}$ such that $\predright(y',y'_{+1},w) \in F_k$, as such an element is created by the application of Rule~\ref{rule:regular-left}. The same application create the atom $\predb(y'_{+1},w)$, which is consistent with the fact that $\conf(A_p)$ contains a $b$ in the first cell at the right of the head; cells further to the right are treated similarly to cells to the left, which are necessarily created by Rule~\ref{rule:copyright}.
  \item the only way to derive an atom of the shape $\predend(x',x'')$ is to apply Rule~\ref{rule:end}, which can only be done after $n-p$ rule applications of Rule~\ref{rule:copyright}, yielding a path of length $n-p+2$ from position $p-1$ of the current configuration, which fulfills the condition of Definition~\ref{def:consistency} (remember that the length of $\conf(A)$ is incremented by $1$ with respect to the length of $\conf(A_p)$).
 \end{itemize}
%
%
%

\end{proof}
\end{toappendix}

\begin{lemmarep}\label{lem:RM_simulates}
	For every configuration $\rho$, if the restricted chase does not terminate on $\langle\aruleset_M,D_\rho\rangle$ then there exists a run of $M$ on $\rho$ which visits \qloop infinitely many times.
\end{lemmarep}

\begin{proof}
 We consider the sequence of states atoms $(A_n)_{n\in\mathbb{N}}$ provided by Lemma~\ref{lem:infinite_sequence_of_state_atoms}, and the sequence $(\conf(A_n))_{n\in \mathbb{N}}$. 
 
 \begin{itemize}
\item $\conf(A_0)$ is the starting configuration of $M$ on $\varepsilon$, and thus a run of $M$ on that configuration;
\item if $(A_n)_{n\in\mathbb{N}}$ is not a run, there exists a smallest $j \in \mathbb{N}$ such that $(\conf(A_n))_{1 \leq n \leq j}$ is not a run. $\conf(A_{j-1})$ is consistent with the set of atoms associated with $A_{j-1}$ by Proposition~\ref{prop:consistency}. Hence $\conf(A_j)$ is obtained by applying the transition correponding to the rule creating $A_j$, and thus $(\conf(A_n))_{1 \leq n \leq j}$ is a run, which leads to a contradiction.
 \end{itemize}

\end{proof}

Lemmas~\ref{lemma:TM_to_chase} and \ref{lem:RM_simulates} directly imply Proposition~\ref{prop:kb-reduction-correct}, and hence the correctness of the reduction.

\section{Rule set termination}
\label{section:rule-set-termination}

\begin{theorem}\label{thm:rule_set_termination_hardness}
  $\TRR{\forall}$ is $\Pi_1^1$-complete.
\end{theorem}

The theorem immediately follows from the upcoming \cref{lem:rs-termination-mem} and \cref{lem:rs-termination-hard}.

\begin{lemma}\label{lem:rs-termination-mem}
	Deciding membership in $\TRR{\forall}$ is in $\Pi_1^1$.
\end{lemma}
\begin{proof}
  We reduce to universal non-recurrence through \qloop.
  More precisely, we show that a rule set $\aruleset$ is in $\TRR{\forall}$ if and only if $\mathcal{M}_\aruleset$
  from \cref{def:termination-aux-machine} is universally non-recurring through \qloop.
  
  If $\aruleset$ is in $\TRR{\forall}$, then $\langle \aruleset, \db \rangle$ is in $\TKR{\forall}$ for each \db.
  Hence, by \cref{lem:termination-aux-machine}, $\mathcal{M}_\aruleset$ is non-recurring on every input that is the encoding of some database \db. On inputs that are not encodings of databases, $\mathcal{M}_\aruleset$ halts immediately by \cref{def:termination-aux-machine}.
  Therefore, $\mathcal{M}_\aruleset$ is universally non-recurring.
  
  If $\mathcal{M}_\aruleset$ is universally non-recurring through \qloop, then, in particular, $\mathcal{M}_\aruleset$ is non-recurring through \qloop on every input that is the encoding of a database. Hence, by \cref{lem:termination-aux-machine}, each restricted chase sequence for each knowledge base with $\aruleset$ is finite. Therefore, $\aruleset$ is in $\TRR{\forall}$.
\end{proof}

\begin{lemma}\label{lem:rs-termination-hard}
	$\TRR{\forall}$ membership is $\Pi_1^1$-hard.
\end{lemma}

The rest of the section is dedicated to proving this lemma, by reducing robust non-recurrence through \qloop to rule set termination. In fact, the reduction is very similar to the one we use for knowledge base termination: to a machine $M$, we associate the rule set $\aruleset_M$, which will belong to $\TRR{\forall}$ if and only if $M$ is robust non-recurring through \qloop. 
The direct implication follows from \cref{lem:RM_simulates} by contrapositive: 
if a Turing machine $M$ is \emph{not} robust non-recurring through \qloop, then there is a configuration $\rho$ such that $M$ visits \qloop infinitely many times from $\rho$. Then, by \cref{lem:RM_simulates}, the restricted chase does not terminate on $\angles{\aruleset_M,D_\rho}$, and thus $\aruleset_M\notin\TRR{\forall}$. The other direction requires more work. Consider a Turing machine $M$, and assume that there is some database $D$ such that the restricted chase does not terminate on $\angles{\aruleset_M,D}$. We then show that $M$ is not robust non-recurring through \qloop.

Since the restricted chase does not terminate on $\angles{\aruleset_M,D}$, there is an infinite chase sequence from this knowledge base. We use $F$ to denote its result. As in \cref{section:kb-termination}, by \cref{lem:infinite_sequence_of_state_atoms}, $F$ contains an infinite sequence of state atoms $\seqstateatoms=(A_n)_{n\in\nat}$ such that $A_0\in D$, $A_n\prec A_{n+1}$ for all $n\in\nat$, and there are infinitely many integers $i$ such that $A_i$ is a $\predqloop$-atom.

In the knowledge base case, we had control over the database as part of the knowledge base, which meant that we could start from a ``well-formed'' database (in the sense that it encodes a single start configuration). This allowed us to extract the \emph{unique} configuration associated with a state atom. However, in the rule set case, the database $D$ leading to non-termination is arbitrary and can contain any kind of structure, as highlighted by the following example.

{ 
\newcommand*{\outxsep}{4}
\newcommand*{\inxsep}{0.6}
\newcommand*{\ysep}{0.9}

\newcommand*{\xlegend}{\xa-2*\inxsep}

\newcommand*{\xb}{0}
\newcommand*{\xa}{\xb-\inxsep}
\newcommand*{\xc}{\xb+\inxsep}

\newcommand*{\xe}{\xb+\outxsep}
\newcommand*{\xd}{\xe-\inxsep}
\newcommand*{\xf}{\xe+\inxsep}

\newcommand*{\xh}{\xe+\outxsep}
\newcommand*{\xg}{\xh-\inxsep}
\newcommand*{\xj}{\xh+\inxsep}

\newcommand*{\xl}{\xh+4*\inxsep}
\newcommand*{\xk}{\xl-\inxsep}
\newcommand*{\xm}{\xl+\inxsep}

\newcommand*{\xarra}{\xb+0.5*\outxsep}
\newcommand*{\xarrb}{\xe+0.5*\outxsep}

\newcommand*{\ya}{4*\ysep}
\newcommand*{\yb}{3*\ysep}
\newcommand*{\yc}{2*\ysep}
\newcommand*{\yd}{1*\ysep}
\newcommand*{\ye}{0}

\begin{figure}
\begin{tikzpicture}
[term/.style={draw, fill=black, circle, inner sep=1pt},
rightarrow/.style={->, thick}]

	\path[dashed, opacity=0.2] (\xlegend+0.5, \ya) edge (\xm+0.5, \ya);
	\path[dashed, opacity=0.2] (\xlegend+0.5, \yb) edge (\xm+0.5, \yb);
	\path[dashed, opacity=0.2] (\xlegend+0.5, \yc) edge (\xm+0.5, \yc);
	\path[dashed, opacity=0.2] (\xlegend+0.5, \yd) edge (\xm+0.5, \yd);
	\path[dashed, opacity=0.2] (\xlegend+0.5, \ye) edge (\xm+0.5, \ye);

	\node[term] (a-1) at (\xb, \ya) {};
	\node[term] (b1-1) at (\xa, \yb) {};
	\node[term] (b2-1) at (\xc, \yb) {};
	\node[term] (c-1) at (\xb, \yc) {};
	\node[term] (d-1) at (\xb, \yd) {};
	\node[term] (e-1) at (\xb, \ye) {};
	\node[above right] at 	(a-1) {\pn{1}};
	\node[above left] at 	(b1-1) {\pn{1}};
	\node[above right] at 	(b2-1) {\pn{0}};
	\node[below right] at 		(c-1) {\pn{0},\predqstart};
	\node[right] at 	(d-1) {\pn{1}};
	\node[above right] at 	(e-1) {\pn{1},\pn{\blank}};

	\path[rightarrow] (a-1) edge (b1-1);
	\path[rightarrow] (a-1) edge (b2-1);
	\path[rightarrow] (b1-1) edge (c-1);
	\path[rightarrow] (b2-1) edge (c-1);
	\path[rightarrow] (c-1) edge (d-1);
	\path[rightarrow] (d-1) edge (e-1);

	\node[term] (a1-2) at (\xd, \ya) {};
	\node[term] (a2-2) at (\xf, \ya) {};
	\node[term] (b1-2) at (\xd, \yb) {};
	\node[term] (b2-2) at (\xf, \yb) {};
	\node[term] (c-2) at (\xe, \yc) {};
	\node[term] (d-2) at (\xe, \yd) {};
	\node[term] (e1-2) at (\xd, \ye) {};
	\node[term] (e2-2) at (\xf, \ye) {};
	\node[above right] at 	(a1-2) {\pn{1}};
	\node[above right] at 	(a2-2) {\pn{1}};
	\node[above right] at 	(b1-2) {\pn{1}};
	\node[above right] at 	(b2-2) {\pn{0}};
	\node[below right] at 		(c-2) {\pn{1}};
	\node[right] at 	(d-2) {\pn{1},\predqloop};
	\node[above left] at 	(e1-2) {\pn{1}};
	\node[above right] at 	(e2-2) {\pn{\blank}};
	
	\path[rightarrow] (a1-2) edge (b1-2);
	\path[rightarrow] (a2-2) edge (b2-2);
	\path[rightarrow] (b1-2) edge (c-2);
	\path[rightarrow] (b2-2) edge (c-2);
	\path[rightarrow] (c-2) edge (d-2);
	\path[rightarrow] (d-2) edge (e1-2);
	\path[rightarrow] (d-2) edge (e2-2);

	\node[term] (a1-3) at (\xg, \ya) {};
	\node[term] (a2-3) at (\xj, \ya) {};
	\node[term] (b1-3) at (\xg, \yb) {};
	\node[term] (b2-3) at (\xj, \yb) {};
	\node[term] (c-3) at (\xh, \yc) {};
	\node[term] (d-3) at (\xh, \yd) {};
	\node[term] (e-3) at (\xh, \ye) {};
	\node[above right] at 	(a1-3) {\pn{1}};
	\node[above right] at 	(a2-3) {\pn{1}};
	\node[above right] at 	(b1-3) {\pn{1}};
	\node[above right] at 	(b2-3) {\pn{0}};
	\node[below right] at 		(c-3) {\pn{1}};
	\node[right] at 	(d-3) {\pn{1}};
	\node[above right] at 	(e-3) {\pn{1},\predqstart};

	\path[rightarrow] (a1-3) edge (b1-3);
	\path[rightarrow] (a2-3) edge (b2-3);
	\path[rightarrow] (b1-3) edge (c-3);
	\path[rightarrow] (b2-3) edge (c-3);
	\path[rightarrow] (c-3) edge (d-3);
	\path[rightarrow] (d-3) edge (e-3);

	\node[term] (a1-4) at (\xk, \ya) {};
	\node[term] (a2-4) at (\xm, \ya) {};
	\node[term] (b1-4) at (\xk, \yb) {};
	\node[term] (b2-4) at (\xm, \yb) {};
	\node[term] (c-4) at (\xl, \yc) {};
	\node[term] (d-4) at (\xl, \yd) {};
	\node[term] (e-4) at (\xl, \ye) {};
	\node[above right] at 	(a1-4) {\pn{1}};
	\node[above right] at 	(a2-4) {\pn{1}};
	\node[above right] at 	(b1-4) {\pn{1}};
	\node[above right] at 	(b2-4) {\pn{0}};
	\node[below right] at 		(c-4) {\pn{1}};
	\node[right] at 	(d-4) {\pn{1}};
	\node[above right] at 	(e-4) {\pn{\blank},\predqstart};
	
	\path[rightarrow] (a1-4) edge (b1-4);
	\path[rightarrow] (a2-4) edge (b2-4);
	\path[rightarrow] (b1-4) edge (c-4);
	\path[rightarrow] (b2-4) edge (c-4);
	\path[rightarrow] (c-4) edge (d-4);
	\path[rightarrow] (d-4) edge (e-4);
	
	\node[fill=white] (arra) at (\xarra, \yc) {\LARGE$\Rightarrow$};
	\node[fill=white] (arrb) at (\xarrb, \yc) {\LARGE$\Rightarrow$};
	\node[above=0cm of arra] {\predfuture};
	\node[above=0cm of arrb] {\predfuture};

	\node (lega) at (\xlegend, \ya) {\Large $a$};
	\node (legb) at (\xlegend, \yb) {\Large $b$};
	\node (legc) at (\xlegend, \yc) {\Large $c$};
	\node (legd) at (\xlegend, \yd) {\Large $d$};
	\node (lege) at (\xlegend, \ye) {\Large $e$};

\end{tikzpicture}
\caption{First three steps of the restricted chase from $\angles{\setR_M,D}$ as defined in \cref{ex:rule_set_term_non_standard_db}. The predicate $\predfuture$ and the brakes are not represented for the sake of readability, but terms are connected through the future predicate to an element on the same line at the previous step. Unlabeled arrows represent $\predright$-atoms.}
\label{fig:rule_set_term_non_standard_db}
\end{figure}
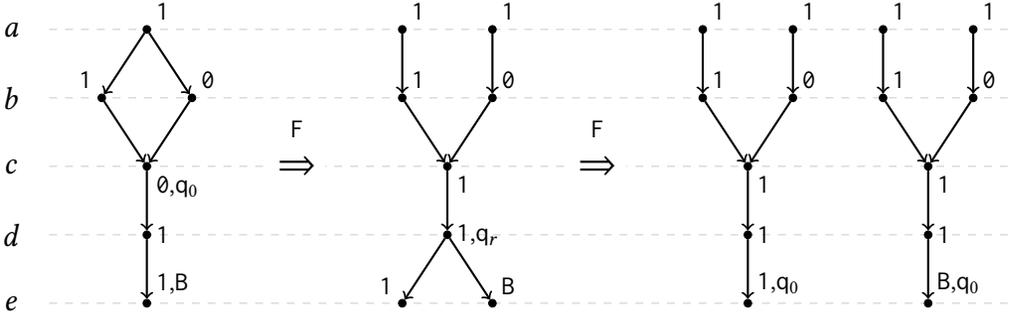
} 

\begin{example}\label{ex:rule_set_term_non_standard_db}
    Consider a Turing machine $M$ that moves to the right in every step, writing $1$ regardless of the symbol it reads. It alternates between its start state $\qinit$ and the designated state $\qloop$. Now, consider the database depicted on the left side of \cref{fig:rule_set_term_non_standard_db}, which contains the atoms $\predright(a, b_1, w)$, $\predright(a, b_2, w)$, $\predright(b_1, c, w)$, $\predright(b_2, c, w)$, $\predright(c, d, w)$, $\predright(d, e, w)$, $\predqstart(c, w)$, $\predone(a, w)$, $\predone(b_1, w)$, $\predzero(b_2, w)$, $\predzero(c, w)$, $\predone(d, w)$, $\predone(e, w)$, $\pn{\blank}(e, w)$, and $\brake(x, w)$ for all $x \in \{a, b_1, b_2, c, d, e\}$.
    This database represents four different configurations, each with a tape of size 5, the start state $\predqstart$, and the head positioned at the third cell. These configurations correspond to tapes with contents $11011$, $10011$, $1101\blank$, and $1001\blank$.
    
    As the simulation progresses, these configurations evolve simultaneously, creating new structures shown in the middle and right of \cref{fig:rule_set_term_non_standard_db}. Notice how term $e$ has two successors through the $\predfuture$ predicate, one for each symbol atom it belongs to. Furthermore, when the head encounters a branching structure, it splits into two, as observed in the third step of the simulation.
    In a sense, if the machine simulation is able to perform steps on the database at all, then it will gradually ``heal'' the structure step by step towards proper encodings of machine configurations.
\end{example}

As highlighted by this example, the structure of the set of atoms connected to a state atom not present in the database is specific: it is the union of two trees rooted in the state atom. The first has arrows going towards the state atom, and the second one has arrows going away from the state atom. In fact, this structure represent the set of paths in the initial database (after the appropriate number of steps of simulation), which we coin a bow tie, due to its shape. 

\begin{definition}\label{def:bowtie}
	The inverse $E^{-}$ of a binary relation $E$ is the relation defined by $(x,y)\in E^-$ if and only if $(y,x)\in E$. In a directed graph $G=(V,E)$ we denote with $V_x^{-y}$ the connected component\footnote{We consider here weakly connected components; a weakly connected component in a directed graph is a maximal subgraph such that there is an undirected path between any two vertices.} of $x$ in the subgraph induced by $V\setminus\set{y}$ on $G$, for any two vertices $x$ and $y$. A \emph{bow tie} is a graph $(V,E)$ with two distinguished vertices $x$ and $y$ that has the following properties:
	\begin{enumerate}
		\item $(x,y)\in E$;
		\item The subgraph induced by $V_x^{-y}$ on $(V,E^{-})$ is a directed tree rooted in $x$;
		\item The subgraph induced by $V_y^{-x}$ on $(V,E)$ is a directed tree rooted in $y$;
		\item The sets $V_x^{-y}$ and $V_y^{-x}$ form a partition of $V$; that is, they are disjoint and $V=V_x^{-y}\cup V_y^{-x}$.
	\end{enumerate}
	The edge $(x,y)$ is called the \emph{center} of the bow tie, and the sets $V_x^{-y}$ and $V_y^{-x}$ are called the \emph{left} and \emph{right} parts of the bow tie, respectively.
\end{definition}

In the following, we denote with $\semterms(F)$ (for semantically meaningful terms) the set of all the terms in $F$, except the brakes (which appear in the last position of atoms). We also define $E_R$ as the relation over $\semterms(F)$ such that $(x,y)\in E_R$ if and only if there is $w$ such that $R(x,y,w)\in F$. For all state atoms $A=\predq(x,w)$ generated during the chase, we denote the connected component of $x$ in the graph $(\semterms(F),E_R)$ with $\setbowtie(A)$. The following lemma explains how this bow tie structure is generated at each step.

\begin{toappendix}
	The following lemmas are used in later proofs of the section.
	\begin{lemma}\label{lem:future_is_inverse_functional_on_nulls}
		For all databases $D$, and every $F$ result of a chase sequence for $\langle\aruleset_M,D\rangle$, if the atoms $\predfuture(x,z,w)$ and $\predfuture(y,z,w)$ are both in $F$ and $z$ is a null, then $x=y$.
	\end{lemma}
	\begin{proof}
		This result follows from the fact that whenever an $\predfuture$-atom appears in the head of a rule, it contains an existentially quantified variable in second position, and no two $\predfuture$-atoms contain this variable in second position. Thus, if $z$ is a null and $x$ and $y$ are different, the atoms $\predfuture(x,z,w)$ and $\predfuture(y,z,w)$ must have been generated by two rule applications, which both introduce $z$, which is impossible.
	\end{proof}

	\begin{lemma}\label{lem:letter_atom_unique_on_nulls}
		For all databases $D$, and every $F$ result of a chase sequence for $\langle\aruleset_M,D\rangle$, for each null $y$ in $\semterms(F\setminus D)$, there are a unique $w$ and a unique $a\in \Gamma \cup \{\blank\}$ such that $\predname{a}(y,w)\in F$.
	\end{lemma}
	\begin{proof}
		Whenever there is an existentially quantified variable $x$ in the head of a rule in $\aruleset_M\setminus\set{\ref{rule:brake}}$, it appears in a unique atom of the form $\predname{a}(x,w)$ in the same head. In addition, all the atoms of the same form in heads of rules feature an existentially quantified variable in first position (except for \ref{rule:brake}, which feature a brake). Thus, when the null $y$ is introduced in the chase, there is a unique atom $\predname{a}(y,w)$ introduced along with it (hence implying existence), and no other rule can introduce an atom of the same form (hence implying uniqueness).
	\end{proof}
\end{toappendix}

\begin{lemmarep}\label{lem:bowtie}
	For all database $D$, and every $F$ result of a chase sequence for $\langle\aruleset_M,D\rangle$, the graph $\setbowtie(A)$ is a finite bow tie for all state atoms $A\in F\setminus D$.
	In addition:
	\begin{itemize}
		\item The center of the bow tie is the atom generated along with $A$, by rule \ref{rule:regular-left}, \ref{rule:regular-right}, \ref{rule:init-left} or \ref{rule:init-right};
		\item all the atoms in the left part of the bow tie are generated by rule \ref{rule:copyleft};
		\item all the atoms in the right part of the bow tie are generated by rule \ref{rule:copyright}, except possibly the end of a maximal path, which may have been generated by rule \ref{rule:end}.
	\end{itemize}
\end{lemmarep}
\begin{proofsketch}
	This proof relies on an analysis of how $\predright$-atoms are generated during the chase. All the rules that generate $\predright$-atoms (over non-brake terms) generate $\predright$-atoms containing at least one existentially quantified variable. Three cases occur:
	\begin{itemize}
		\item Rules \ref{rule:regular-left}, \ref{rule:regular-right}, \ref{rule:init-left} and \ref{rule:init-right} generate an $\predright$-atom $\predright(u,v,w)$ where $u$ and $v$ are both existentially quantified.
		\item Rule \ref{rule:copyleft} generates an $\predright$-atom $\predright(u,v,w)$ where $u$ is existentially quantified and $v$ is a frontier variable.
		\item Rules \ref{rule:copyright} and \ref{rule:end} generate an $\predright$-atom $\predright(u,v,w)$ where $u$ is a frontier variable and $v$ is existentially quantified.
	\end{itemize}

	Thus, all connected components are generated by a rule of the first kind, and then extended to the left by a rule of the second kind, and to the right by a rule of the third kind. Since no rule can generate an atom $\predright(u,v,w)$ where $u$ and $v$ are both frontier variable (assuming $u$ and $v$ are not brakes), this yields the wanted structure.
	Finiteness is guaranteed by the emergency brakes.
\end{proofsketch}
\begin{proof}
	First, notice that all the rules that generate $\predright$-atoms (over non-brake terms) generate $\predright$-atoms containing at least one existentially quantified variable. Three cases occur:
	\begin{itemize}
		\item Rules \ref{rule:regular-left}, \ref{rule:regular-right}, \ref{rule:init-left} and \ref{rule:init-right} generate an $\predright$-atom $\predright(u,v,w)$ where $u$ and $v$ are both existentially quantified.
		\item Rule \ref{rule:copyleft} generates an $\predright$-atom $\predright(u,v,w)$ where $u$ is existentially quantified and $v$ is a frontier variable.
		\item Rules \ref{rule:copyright} and \ref{rule:end} generate an $\predright$-atom $\predright(u,v,w)$ where $u$ is a frontier variable and $v$ is existentially quantified.
	\end{itemize}

	Thus, no connected component can contain two $\predright$-atoms that are generated using a rule among rules \ref{rule:regular-left}, \ref{rule:regular-right}, \ref{rule:init-left} and \ref{rule:init-right}. Indeed, these rules create a new connected component, and to connect two connected components, we need a rule generating an $\predright$-atom $\predright(u,v,w)$ where $u$ and $v$ are both frontier variables, which is not the case with this rule set. This also implies that $(\setbowtie(A), E_R)$ is acyclic, even when seen as an undirected graph, for the same reason.

	Thus, since $A=\predq(x,w)$ is generated by a rule among \ref{rule:regular-left}, \ref{rule:regular-right}, \ref{rule:init-left} and \ref{rule:init-right} along with an $\predright$-atom, all the other atoms in the connected component of $x$ must have been generated by \ref{rule:copyleft}, \ref{rule:copyright} or \ref{rule:end}. We assume here that $A$ was generated by rule \ref{rule:regular-left} or \ref{rule:init-left}, as the other cases are symmetric. Then, $A$ is generated along the atom $\predright(x,y,w)$, which will be the center of our bow tie, and atoms $\predcopyleft(x,w)$ and $\predcopyright(y,w)$.

	We then consider the sets $\setbowtie(A)_{x}^{-y}$ and $\setbowtie(A)_{y}^{-x}$ as defined in \cref{def:bowtie}. First, as mentioned before, the undirected graph induced by $(\setbowtie(A),E_R)$ is acyclic and connected, so these sets form a partition of $\setbowtie(A)$. Thus, it only remains to show that the subgraphs induced by $\setbowtie(A)_x^{-y}$ on $(\setbowtie(A),E_R^-)$ and by $\setbowtie(A)_y^{-x}$ on $(\setbowtie(A),E_R)$ are trees. Again, since both proofs are similar, we only prove it for the second graph.

	A directed tree is an acyclic and connected graph such that each vertex has in-degree at most one. Since $(\setbowtie(A),E_R)$ is acyclic, the subgraph induced by $\setbowtie(A)_y^{-x}$ is acyclic too, and as it is a connected component, it is connected. Thus, it only remains to show that each term in $\setbowtie(A)_y^{-x}$ has an in-degree of at most one. Our previous analysis of the rules entails that only a term $t$ such that $\predcopyleft(t,w)\in F$ can have an in-degree greater than one. Indeed, the only rule that can increase the in-degree of an existing element is rule \ref{rule:copyleft}, which requires this atom in its body. We thus show that there is no $t$ in $\setbowtie(A)_y^{-x}$ such that $\predcopyleft(t,w)\in F$.

	Only two kinds of rules can generate $\predcopyleft$-atoms (over non-brakes), which are transition rules (\ref{rule:regular-left}, \ref{rule:regular-right}, \ref{rule:init-left} and \ref{rule:init-right}), and the rule \ref{rule:copyleft}. All these rules generate atoms of the form $\predcopyleft(u,w)$ where $u$ is existentially quantified. As stated before, in $\setbowtie(A)$, only the atom $\predright(x,y,w)$ has been generated using a transition rule, and every other $\predright$-atom has been generated using \ref{rule:copyleft} or \ref{rule:copyright}. Now, for a contradiction, assume that $t$ is the first term of $\setbowtie(A)_y^{-x}$ introduced during the chase such that $\predcopyleft(t,w)\in F$. Since the trigger generating $\predright(x,y,w)$ only generates $\predcopyleft(x,w)$, and $x\notin\setbowtie(A)_y^{-x}$, the term $t$ has been generated by rule \ref{rule:copyleft}. This means that there is a term $u\in\setbowtie(A)_y^{-x}$ such that $\predcopyleft(u,w)\in F$ before $\predcopyleft(t,w)$ is introduced, which contradicts our hypothesis. Note that $u$ does have to be in $\setbowtie(A)_y^{-x}$, since otherwise, $t\notin\setbowtie(A)_y^{-x}$, as no rule can connect two disjoint connected components.

	Thus, there is no $\predcopyleft$-atom over a term in $\setbowtie(A)_y^{-x}$, meaning that $(\setbowtie(A)_y^{-x}, E_R)$ is a tree. As mentioned before, an analog line of reasoning can be used to show that $(\setbowtie(A)_x^{-y}, E_R^-)$ is also a tree, so $\setbowtie(A)$ is indeed a bow tie. Note also that since no $\predcopyleft$-atom over a term in $\setbowtie(A)_y^{-x}$, all the $\predright$-atoms of the right part of the bow tie must have been generated by rule \ref{rule:copyright} or \ref{rule:end}. However, rule \ref{rule:end} generates a new null $y$ such that $\predcopyright(y,w)\notin F$ (by the same description as previously), and both rules \ref{rule:copyright} and \ref{rule:end} require an atom of this form to extend a path. Thus, if an $\predright$-atom is generated using rule \ref{rule:end}, it is necessarily the end of a maximal path.
	
	It remains to show that $\setbowtie(A)$ is finite: 
	if $F$ is finite, then so is $A$ and therefore $\setbowtie(A)$.
	Otherwise, note that all atoms in $\setbowtie(A)$ are associated with the same brake $w$.
	Then, by \cref{lem:finitely_many_atoms_per_brake}, $\setbowtie(A)$ must be finite.
\end{proof}

We now have a bit of structure to work with. Let us give a bit of intuition before concluding the proof. We have considered an infinite sequence $\seqstateatoms = (A_n)_{n\in\nat}$ of state atoms, with $A_0\in D$ and $A_n\prec A_{n+1}$ for all $n\in\nat$, and we have just shown that to each state atom (not in $D$) is attached a bow tie structure. As mentioned before, the bow tie $\setbowtie(A_n)$ 
consists in a set of (non-disjoint) paths that represent configurations that can be obtained from a configuration containing $A_0$ in the database, after $n$ steps of simulation. 
In addition, \cref{lem:bowtie} shows how each of these paths is constructed using a path from $\setbowtie(A_{n-1})$. We also have seen in \cref{ex:rule_set_term_non_standard_db} that a bow tie can get split. From these two facts we get that the number of configurations represented by $\setbowtie(A_n)$ decreases as $n$ grows. Since this number is an integer, and each bow tie represents at least one configuration, this sequence will be stationnary at some point $N$. At this point, we know that each of the configurations represented by $\setbowtie(A_N)$ visits $\qloop$ infinitely many time. Thus, we pick such a configuration $\rho$, and we show that the restricted chase does not terminate on $\angles{\aruleset_M,D_\rho}$, which is enough to conclude the proof by \cref{lem:RM_simulates}. We then formalize this argument.

\begin{definition}\label{def:configs_A_n}
	The \emph{set of configurations} $\configs(A_n)$ \emph{associated to a state atom} $A_n=q(x,w)\in\seqstateatoms$, with $n>0$, is the set whose elements are the sets
	\[\set{A_n}\cup\bigcup_{i\leq m}\brake(x_i,w)\cup\setst{\pn{P}(y_1,\dots,y_k,w)\in F}{P\in\set{\predright,\pn{0},\pn{1},\pn{\blank},\predend} \text{ and } \forall i,\ y_i\in\set{x_1,\dots,x_m}}\]
	for all maximal paths $(x_1,\dots,x_m)$ in $\setbowtie(A_n)$.
\end{definition}

\begin{toappendix}
Recall that $\seqstateatoms = (A_n)$ is the sequence of state atoms provided by Lemma~\ref{lem:infinite_sequence_of_state_atoms}.
\end{toappendix}

\begin{lemmarep}\label{lem:configs_is_decreasing}
	For all $n > 0$, $\configs(A_n)$ is finite, non-empty, and each of its elements homomorphically embeds into $D_\rho$ for some configuration $\rho$. Also, there is an injective function $\predfun_n$ from $\configs(A_{n+1})$ to $\configs(A_n)$ such that $S\in\configs(A_{n+1})$ can be generated using only atoms in $\predfun_n(S)$.
\end{lemmarep}
\begin{proofsketch}
	To each set $S\in\configs(A_{n+1})$ we can associate a configuration $\rho$ and a path $p$ in $\setbowtie(A_{n+1})$. We then define $\predfun_n(S)$ as the set of atoms that was used to generate it, which is not hard: its associated configuration is an extension of a configuration that yields $\rho$, and its associated path is connected through the $\predfuture$-predicate to $p$. To show injectivity of $\predfun_n$, we then rely on a lemma stating that if $\predfuture(x,z,w)$ and $\predfuture(y,z,w)$ are both in $F$, then $x=y$.
\end{proofsketch}
\begin{proof}
	\textbf{Non-emptiness and finiteness.} Non-emptiness and finiteness of $\configs(A_n)$ follow from \cref{lem:bowtie}, since a finite bow tie has a finite non-zero amount of maximal paths.

	\noindent\textbf{The elements of $\configs(A_n)$ embed into some $D_\rho$.} We then consider an element $S$ of $\configs(A_n)$, and $(x_1,\dots,x_n)$ the path associated with it. Also, let $A_n=q(x,w)$. First, since $(x_1,\dots,x_n)$ is a path in $(\setbowtie(A_n),E_R)$, for all $i$, the atom $\predright(x_i,x_{i+1},w)$ is in $S$ for all $i$, and these are all the $\predright$-atoms in $S$ by \cref{lem:bowtie}. Then, by \cref{lem:letter_atom_unique_on_nulls}, there is a unique atom $\pn{a_i}(x_i,w)$ for all $i\leq n$. In addition, since all the maximal paths in a bow tie go through its center, there is some $p$ such that $x_p=x$. We thus define the configuration $\rho=\langle n,(a_i)_{i\leq n},p,q\rangle$.

	By mapping $x_i$ to $c_i$ for all $i$, and $w$ to $w_1$, we get that $S$ and $D_\rho$ are isomorphic, except for the $\predend$-atoms. However, as per the last item of \cref{lem:bowtie}, the only position that can have $\predend(x_i,w)$ is the end of a maximal path (since this kind of atoms can only be generated by rule \ref{rule:end} over non-brakes). Thus, the only possible $\predend$-atom in $S$ is $\predend(x_n,w)$, which has a counterpart in $D_\rho$. Thus, $S$ homomorphically embeds into $D_\rho$.

	\noindent\textbf{Construction of $\predfun_n$.} We then construct the function $\predfun_n$. Let $A_n=\predq(x,w)$ and $A_{n+1}=\predqp(y,w')$. First notice that the rule that generates $A_{n+1}$ in the chase is among \ref{rule:regular-left}, \ref{rule:regular-right}, \ref{rule:init-left} and \ref{rule:init-right}. Then, there are some atoms $\predfuture(x,z,w')$, and $\predright(z,y,w')$ or $\predright(y,z,w')$ depending on the direction of the transition. We then assume that the transition is to the right, as the left case is analogous.

	Consider a set $S\in\configs(A_{n+1})$, $(y_1,\dots,y_k)$ the associated path, and $\rho'=\langle k,(b_1,\dots,b_k),p',q'\rangle$ the associated configuration, as defined earlier in the proof. Then, let $\predfun_n(S)$ be one of the sets in $\configs(A_n)$ with associated path $(x_1,\dots,x_m)$ and configuration $\rho=\langle m,(a_1,\dots,a_m),p,q\rangle$ such that there is an integer $l$ such that
	\begin{itemize}
		\item for all $i<k$, we have $\predfuture(x_{i+l},y_i,w')\in F$;
		\item if $\predend(x_k,w')\in S$, then $\predend(y_{k+l-1})\in F$, and otherwise $\predfuture(x_{k+l},y_k,w')\in F$;
		\item for all $i\neq p'-1$, $b_i=a_{i+l}$;
		\item $p'+l=p+1$.
	\end{itemize}

	\noindent\textbf{The function $\predfun_n$ is well-defined.} By definition of $S$ and its associated path and configuration, there must be some atoms $\predright(y_i,y_{i+1},w')$ and $\predname{b_i}(y_i,w')$ for all $i$, with $y=y_{p'}$. By \cref{lem:bowtie}, $\predright(y_{p'-1},y_{p'})$ has been generated by rule \ref{rule:regular-right} or \ref{rule:init-right} along with $A_{n+1}$, and $\predright(x_{k-1},x_k)$ may have been generated by rule \ref{rule:end} or \ref{rule:copyright}. Other than that, all the $\predright$-atoms in the path $y_1,\dots,y_k$ have been generated by rules \ref{rule:copyleft} and \ref{rule:copyright}. We then show that there is a path $x'_1,\dots,x'_{k'}$ such that for all $i<k$, $\predfuture(x'_i,y_i,w)\in F$, for all $i\neq p'-1$, $\predname{b_i}(x'_i,w)\in F$, and either $\predend(x'_{k-1},w)\in F$ (and $k'=k-1$) or $\predfuture(x'_k,y_k,w)\in F$ (and $k'=k$), depending on whether $\predend(y_k,w)\in S$ or not.

	First, since the atom $A_{n+1}$ has been generated by rule \ref{rule:regular-right} or \ref{rule:init-right}, there must be a term $z$ and some atoms $\predright(x,z,w)$, $\predright(y_{p'-1},y,w)$, $\predfuture(x,y_{p'-1},w)$ and $\predfuture(z,y,w)$ in $F$. Thus, let $x'_{p'-1}=x$ and $x'_{p'}=z$. We will then extend this path in both directions to construct $x'_1,\dots,x'_{k'}$.

	If the path has been extended up to $x'_{p'+i}$ for some $i<k-1-p'$, we then extend it to $x'_{p'+i+1}$. As mentioned before, the atom $\predright(y_{p'+i},y_{p'+i+1},w)$ has been generated by rule \ref{rule:copyright} (since $p'+i<k$). Thus, there must be some terms $z,t$ and atoms $\predright(z,t,w)$, $\predfuture(z,y_{p'+i},w)$, $\predfuture(t,y_{p'+i+1},w)$ and $\predname{b_i}(t,w)$ in $F$. By \cref{lem:future_is_inverse_functional_on_nulls}, we then have $z=x'_{p'+i}$, since both $\predfuture(z,y_{p'+i},w)$ and $\predfuture(x'_{p'+i},y_{p'+i},w)$ are present in $F$. We thus set $t=x'_{p'+i+1}$. The same reasoning lets us extend the path to $x'_{p'-i-1}$ provided we have extended it to $x'_{p'-i}$, using the left copy rule instead of the right copy.

	We now treat the case where $i=k-1-p$. If $\predend(y_k,w)\notin S$, then $\predright(y_{k-1},y_k,w)$ has been generated by rule \ref{rule:copyright}, so the same reasoning as before applies, and $k'=k$. Otherwise, $\predright(y_{k-1},y_k,w)$ ha been introduced by rule \ref{rule:end}, meaning that there are some term $z$ and atoms $\predend(z,w)$ and $\predfuture(z,y_{k-1})$ in $F$. Thus, since both $\predfuture(z,y_{k-1})$ and $\predfuture(x'_{k-1},y_{k-1})$ are in $F$, by \cref{lem:future_is_inverse_functional_on_nulls}, $z=y_{k-1}$, and we have the atom $\predend(y_{k-1},w)$ in $F$ as promised, and $k'=k-1$.

	We thus have a path $x'_1,\dots,x'_k$ in $\setbowtie(A_n)$ as described before. However, this path does not define an element of $\configs(A_n)$, since it is not maximal. Thus, consider any maximal path $x_1,\dots,x_m$ in $\setbowtie(A_n)$ that extends $x'_1,\dots,x'_k$, $\predfun_n(S)$ the corresponding set in $\configs(A_n)$, $\langle m,(a_1,\dots,a_m),p,q\rangle$ the corresponding configuration, and let $l$ be the integer such that $x_{l+1}=x'_1$. Then, by definition of $(x'_1,\dots,x'_{k'})$, the first two points of the definition of $\predfun_n(S)$ hold. Then, since $A_n=\predq(x'_{p'-1},w)$, and $x'_{p'-1}=x_{p'-1+l}$, we have $p=p'-1+l$, so $p+1=p'+l$. In addition, since for all $i\neq p'-1$, $\predname{b_i}(x_{i+l},w)\in F$, we have $a_{i+l}=b_i$. Thus, there is indeed a set in $\configs(A_n)$ that fits the definition of $\predfun_n(S)$. Note however that this path is not necessarily unique, but we only need an injective function, so this is fine.

	\noindent\textbf{The set $\predfun_n(S)$ is enough to generate $S$.} First note that all the rule applications described earlier suffice to generate $S$. It is then enough to notice that all the atoms in the support of the mentioned triggers are present in $\predfun_n(S)$, or generated during the application of the previous triggers.

	\noindent\textbf{Injectivity of $\predfun_n$.} \sloppypar{Consider two sets $S_1$ with associated path $(y_1,\dots,y_{k_1})$ and configuration $\langle k_1,(b_1,\dots,b_{k_1}),p_1,q'\rangle$, and $S_2$ with associated path $(y'_1,\dots,y'_{k_2})$ and configuration $\langle k_2,(b'_1,\dots,b'_{k_2}),p_2,q'\rangle$, such that $\predfun_n(S_1)=\predfun_n(S_2)=S'$, and $S'$ has path $(x_1,\dots,x_m)$ and configuration $\langle m,(a_1,\dots,a_m),p,q\rangle$. Thus, there must be some $l_1$ and $l_2$ such that:}
	\begin{itemize}
		\item for all $i$, $\predfuture(x_{i+l_1},y_i,w')\in F$ and $\predfuture(x_{i+l_2},y'_i,w')\in F$;
		\item for all $i\neq p_1-1$, $b_i=a_{i+l_1}$ and for all $i\neq p_2-1$, $b'_i=a_{i+l_2}$;
		\item $p_1+l_1=p+1=p_2+l_2$.
	\end{itemize}
	Assume w.l.o.g. that $l_1\geq l_2$, and let $d=l_1-l_2$. We then get that $p_2=p_1+d$, and $b_i=a_{i+l_1}=a_{i+d+l_2}=b'_{i+d}$, for all $i\neq p_1$. We then show that for all $i$ such that $1\leq i\leq k_1$ and $1\leq i+d \leq k_2$, we have $y_i=y'_{i+d}$. First, this is true for $i=p_1$, since $y_{p_1}=y=y'_{p_2}$ (where $A_{n+1}=\predqp(y,w')$) and $p_2=p_1+d$. This is also true for $i=p_1-1$, since by definition of a bow tie and \cref{lem:bowtie}, there is only one term $t$ such that $\predright(t,y_{p_1},w')\in F$. We then extend this to all $i$ by induction.

	Assume that $1\leq i+1\leq k_1$ and $1\leq i+1+d \leq k_2$, and that $y_i=y'_{i+d}$ for some $i\geq p_1$ (the case where $i\leq p_1-1$ is similar, using \ref{rule:copyleft} instead of \ref{rule:copyright}). We then show that $y_{i+1}=y'_{i+1+d}$. Both the atoms $\predright(y_i,y_{i+1},w')$ and $\predright(y_i,y'_{i+1+d},w')$ have been generated using rule \ref{rule:copyright}. We then show that the triggers generating these atoms are equal, so these atoms must be equal. The body of rule \ref{rule:copyright} is $\set{\predcopyright(x',w'),\predfuture(x,x',w'),\predright(x,y,w),\predname{b_i}(y,w),\predreal(x),\predreal(x'),\predreal(y)}$. To generate $\predright(y_i,y_{i+1},w')$, $x'$ must be mapped to $y_i$ (and $w'$ to himself). Then, by \cref{lem:future_is_inverse_functional_on_nulls}, each term $v$ can only have one term $u$ such that $\predfuture(u,v,w)\in F$, so $x$ is mapped to $x_{i+l_1}$ and $y$ to $x_{i+1+l_1}$ (and $w$ to himself), since $\predfuture(x_{i+l_1},y_i,w')$ and $\predfuture(x_{i+1+l_1},y_{i+1},w')$. However, we also have $\predfuture(x_{i+l_1},y'_{i+d},w')$ and $\predfuture(x_{i+1+l_1},y'_{i+1+d},w')$, so the triggers generating $\predright(y_i,y_{i+1},w')$ and $\predright(y_i(\rho,\vx),y'_{i+1+d},w')$ are equal, and $y_{i+1}=y'_{i+1+d}$.

	Thus, $l_1=l_2$ and $k_1=k_2$. Indeed, if $l_1>l_2$, then we can extend $y_1,\dots,y_{k_1}$ into a bigger path $y'_1,\dots,y'_d,y_1,\dots,y_{k_1}$, which contradicts its maximality. If $k_1\neq k_2$, then we can extend the shortest path into the longest, also contradicting its maximality. Thus, both paths are equal, and $S_1=S_2$. From this we deduce that $\predfun_n$ is injective.
\end{proof}

Since for all $n$, there is an injective function from $\configs(A_{n+1})$ to $\configs(A_n)$, the sequence $(\size{\configs(A_n)})_{n\in\nat_{>0}}$ is a decreasing sequence of natural numbers, as mentioned before. Thus, there must be some $N\in\nat$ such that for all $n\geq N$, $\size{\configs(A_n)}=\size{\configs(A_N)}>0$. We pick $S_0$ in $\configs(A_N)$, and let $\rho$ be a configuration such that $S_0$ homomorphically embeds into $D_\rho$.

\begin{lemmarep}
	The restricted chase does not terminate on $\angles{\aruleset_M,D_\rho}$.
\end{lemmarep}
\begin{proofsketch}
	Since for all $n\geq N$, $\size{\configs(A_n)}=\size{\configs(A_N)}$, $\predfun_n$ is actually a bijection. We thus define $S_{n+1}$ as $\predfun_{N+n}^{-1}(S_n)$.
	Intuitively, the sequence $(S_n)_{n\in\nat}$ encodes the run of $M$ that visits \qloop infinitely many times from $\rho$. We then construct an infinite chase sequence from $\angles{\aruleset_M,D_\rho}$ such that $S_n$ homomorphically embed in it for all $n$.
\end{proofsketch}
\begin{proof}
	First note that since for all $n\geq N$, $\size{\configs(A_n)}=\size{\configs(A_N)}$, $\predfun_n$ is actually a bijection, since it is injective between sets of equal sizes. It thus has an inverse $\predfun_n^{-1}$. Thus, for all $n\in\nat$, we define $S_{n+1}$ as $\predfun_{N+n}^{-1}(S_n)$.
	Note that we picked $S_0$ from $\configs(A_N)$ and that $S_0$ homomorphically embeds into $D_\rho$.

	We then inductively construct a sequence of derivations $(\der_n)_{n\in\nat}$ such that for all $n\in\nat$, if $S_n$ is over terms $x_1,\dots,x_k,w$, then
	\begin{itemize}
		\item $\der_{n+1}$ extends $\der_n$;
		\item there is a homomorphism $\pi_n$ from $S_n$ to the result $R_n$ of $\der_n$;
		\item if $\predfuture(\pi_n(x_i),y,\pi_n(w))\in\der_n$ for some $i$, then $y=\pi_n(w)$;
		\item $\predreal(\pi_n(w))\notin\der_n$
	\end{itemize}

	First, as stated in \cref{lem:configs_is_decreasing}, $S_0$ embeds in $D_\rho$, so we let $\der_0=D_\rho$, which does fulfill all the conditions above. Then, assume that we have constructed derivation $\der_n$ as described. By \cref{lem:configs_is_decreasing} again, all the atoms in $S_{n+1}$ can be generated using only atoms in $\predfun_n(S_{n+1})=S_n$. We thus extend the derivation $\der_n$ into a derivation $\der'_{n+1}$ with the triggers needed to generate the atoms in $S_{n+1}$, composed with $\pi_n$. All these triggers are applicable since they all create atoms of the form $\predfuture(\pi_n(x_i),y,\pi_n(w))$ and $\predreal(y)$, which are not in the database by the third and forth item. The homomorphism $\pi_{n+1}$ is then defined naturally (the triggers that generate $S_{n+1}$ from $S_n$ were used here to generate new nulls, to which we can map nulls of $S_{n+1}$). Then, if $S_n$ contains an atom of the form $\predqloop(x,w')$, we add the trigger $(\ref{rule:brake},\set{w\to\pi_n(w)})$ at the end of this new derivation, to construct $\der_{n+1}$. The first and second point then follow by design. The third point follows from the fact that the triggers that were used to generate $S_{n+1}$ from $S_n$ do not generate other $\predfuture$-atoms, and the last point from the fact that if $\predqloop(x,w')\in S_n$, then $S_n$ and $S_{n+1}$ use different brakes.

	We now show that the derivation $\der=\bigcup_{n} \der_n$ is fair. First, by \cref{lem:infinite_sequence_of_state_atoms}, there are infinitely many $\predqloop$-atoms in $(A_n)_{n\in\nat}$, and thus infinitely many $n\in\nat$ such that $S_n$ contains a $\predqloop$-atom. Then, notice that whenever we encounter a $\predqloop$-atom in $\der$, we make the previous brake real, blocking any rule application involving the atoms containing it. Thus, for any trigger that is applicable at some step $n$, there is a step $m$ at which the brake that appears in this trigger's support gets real, making this trigger obsolete. Thus, $\der$ is fair, and the restricted chase does not terminate on $\angles{\aruleset_M,D_\rho}$.
\end{proof}

By \cref{lem:RM_simulates}, this means that there is a run of $M$ which visits \qloop infinitely many times, and thus that $M$ is not robust non-recurring through \qloop, concluding the reduction.

\section{An Alternative to Fairness to Simplify Restricted Chase Termination}
\label{section:fairness}

All chase variants can be applied to semi-decide Boolean conjunctive query (BCQ) entailment.
This is the case because, if a KB \kb entails a BCQ \aquery under standard first-order semantics, then every chase sequence of \kb features a fact set that entails \aquery.
Consequently, it suffices to compute an (arbitrarily large) finite prefix of any (arbitrarily chosen) chase sequence of \kb to semi-decide whether \kb entails \aquery.
Note that semi-decidability of BCQ entailment breaks down if we remove the fairness condition from the definition of a chase sequence.
Unfortunately, this condition complicates the problem of universal termination for the restricted chase (see \cref{theorem:kb-completeness,thm:rule_set_termination_hardness}).
To address this situation, we propose an alternative to fairness in the following definition that retains semi-decidability while simplifying the termination problem of the chase (see \cref{theorem:termination-b1r}).


\begin{definition}
\label{definition:breadth-first-chase}
A \emph{breadth-first chase sequence} for a KB $\langle \aruleset, \db \rangle$ is a chase derivation $F_0, F_1, \ldots$ such that, $(\dagger)$ if some $\aruleset$-trigger $\lambda$ is loaded for some $F_i$, then there is some $j \in \{i, \ldots, i+n\}$ such that $\lambda$ is obsolete for $F_j$ and $n$ is the (finite) number of $\aruleset$-triggers that are loaded and not obsolete for $F_i$.
\end{definition}

Note that, since $(\dagger)$ implies fairness as introduced in \cref{definition:chase}, every breadth-first chase sequence is also a chase sequence and we preserve semi-decidability of BCQ entailment.

\begin{definition}
Let \TKRBF{\forall} be the class of all KBs that only admit finite breadth-first chase sequences.
Let \TRRBF{\forall} be the class containing a rule set if $\TKRBF{\forall}$ contains all KBs with this rule set.
\end{definition}

\begin{theorem}
\label{theorem:termination-b1r}
The class \TKRBF{\forall} is in \re, and the class \TRRBF{\forall} is in $\Pi_2^0$.
\end{theorem}
\begin{proof}
To show that \TKRBF{\forall} is in \re, we define a semi-decision procedure, which executes the following instructions on a given input KB $\kb = \langle \aruleset, \db\rangle$:
\begin{enumerate}
\item Initialise the set $\mathcal{P}_1$ of lists of facts that contains the (unary) list $\db$, and a counter $i := 2$.
\item Compute the set $\mathcal{C}_i$ of all chase derivations of length $i$ of \kb that can be obtained by extending a chase derivation in $\mathcal{P}_{i-1}$ with one fact set.
Intuitively, $\mathcal{C}_i$ includes all lists of length $i$ that can be extended into breadth-first chase sequences for \kb.
\item Compute the maximal subset $\mathcal{P}_i$ of $\mathcal{C}_i$ that does not contain a chase derivation $F_1, \ldots, F_i \in \mathcal{C}_i$ if there is some $1 \leq k \leq i$ and some $\aruleset$-trigger $\lambda$ such that $\lambda$ is loaded for $F_k$, the trigger $\lambda$ is not obsolete for $F_i$, and $i - k$ is larger than the number of $\aruleset$-triggers that are loaded and not obsolete for $F_k$.
Intuitively, $\mathcal{P}_i$ filters out prefixes in $\mathcal{C}_i$ that already violate $(\dagger)$.
\item If $\mathcal{P}_i$ is empty, \emph{accept}.
Otherwise, increment $i := i + 1$ and go to 2.
\end{enumerate}
If the procedure accepts, then $\mathcal{P}_i$ is empty for some $i$ and all breadth-first chase sequences of \kb are of length at most $i-1$.
If the procedure loops, then there is an infinite chase derivation $F_0, F_1, \ldots$ of \kb such that $F_0, \ldots, F_{i-1} \in \mathcal{P}_i$ for every $i \geq 1$, which is a breadth-first derivation for \kb.

The class \TRRBF{\forall} is in $\Pi_2^0$ because we can semi-decide if a rule set $\aruleset$ is not in $\TRRBF{\forall}$ using an oracle that solves \TKRBF{\forall}.
We simply enumerate every database \db, use the oracle to check if the KB $\langle \aruleset, \db \rangle$ is in $\TKRBF{\forall}$, and \emph{accept} if this is not the case.
\end{proof}

The previous result holds because the condition $(\dagger)$ is \emph{finitely verifiable}; that is, every infinite chase derivation that does not satisfy this condition has a finite prefix that witnesses this violation.
Note that fairness does not have this property since any finite prefix of any chase derivation can be extended into a (fair) chase sequence.
In fact, we can readily show a version of \cref{theorem:termination-b1r} for any other alternative condition if it is finitely verifiable.
For an example of one such trigger application strategy, consider the one from \cite{DBLP:conf/cade/UrbaniKJDC18}, which is a bit more complex to define than $(\dagger)$ but nevertheless results in a very efficient implementation of the restricted chase.


\section{Open Problems}
\label{section:conclusions}

After settling the general case regarding restricted chase termination and proposing an alternative fairness condition, there are still open challenges.
Namely, what is the undecidability status of the classes \TKR{\forall} and \TRR{\forall} if we only consider single-head rules or only guarded (multi-head) rules?
Note that,
with guarded rules, it is not obvious how to simulate a Turing machine.
For single-head rules, we cannot implement the emergency brake and thus our proofs do not apply. 
Moreover, if we only consider single-head rule sets, we can ignore fairness when determining restricted chase termination because of the ``fairness theorem'' \cite{DBLP:journals/siamcomp/GogaczMP23}: a single-head KB admits an infinite (possibly unfair) chase derivation if and only if admits an infinite (fair) chase sequence.
We think that answers to these problems will help to develop a better understanding for the (restricted) chase overall.

\begin{acks}
  On TU Dresden side, 
  this work is partly supported by 
  Deutsche Forschungsgemeinschaft (DFG, German Research Foundation) in project number 389792660 (TRR 248, \href{https://www.perspicuous-computing.science/}{Center for Perspicuous Systems}),
  by the Bundesministerium für Bildung und Forschung (BMBF, Federal Ministry of Education and Research) in the
  \href{https://www.scads.de/}{Center for Scalable Data Analytics and Artificial Intelligence} (project SCADS25B, \href{https://scads.ai/}{ScaDS.AI}),
  and by Bundesministerium für Bildung und Forschung (BMBF, Federal Ministry of Education and Research) and Deutscher Akademischer Austauschdienst (DAAD, German Academic Exchange Service) in project 57616814 (\href{https://www.secai.org/}{SECAI}, \href{https://www.secai.org/}{School of Embedded and Composite AI}).
\end{acks}

\bibliographystyle{ACM-Reference-Format}
\bibliography{bibliography}

\appendix

\end{document}